\documentclass[reqno, 12pt]{amsart}

\usepackage{amssymb,amsrefs}							
\usepackage[T1]{fontenc}								
\usepackage[english]{babel}							
\usepackage[onehalfspacing]{setspace}					
\usepackage[text={6.6in,8.6in}]{geometry}			
\usepackage{fancyhdr,times,mathtools}					
\usepackage[all]{xy}									
\usepackage{graphicx}		 				  			
\usepackage{pinlabel}									
\usepackage[colorlinks,bookmarks=false]{hyperref}	
\numberwithin{equation}{section}						
\usepackage{caption}									
\usepackage[nameinlink,noabbrev]{cleveref}			
\usepackage[all]{hypcap}								
\usepackage{autonum}

\newtheorem*{main}{Main Theorem}
\newtheorem*{main2}{Main Theorem for the large area limit}
\newtheorem{thm}{Theorem}[section]
\newtheorem{lem}[thm]{Lemma}
\newtheorem{cor}[thm]{Corollary}

\newtheorem{example}[thm]{Example}
\crefname{thm}{Theorem}{Theorems}				
\creflabelformat{thm}{#2{#1}#3}					
\crefname{ineq}{inequality}{inequalities}		
\creflabelformat{ineq}{#2{\upshape(#1)}#3}		
\crefname{diag}{diagram}{diagrams}				
\creflabelformat{diag}{#2{\upshape(#1)}#3} 		

\makeatletter
\renewcommand*{\eqref}[1]{\hyperref[{#1}]{\textup{\tagform@{\ref*{#1}}}}}		
\makeatother

\let\originalleft\left
\let\originalright\right
\renewcommand{\left}{\mathopen{}\mathclose\bgroup\originalleft}
\renewcommand{\right}{\aftergroup\egroup\originalright}

  \def\del{\overline \partial} 		
  \def\vor{\upsilon} 					
  \def\ep{\varepsilon}					
  \def\dist{\mathrm{dist}}			
  \def\Ar{\mathrm{Area}}				
  \def\supp{\mathrm{supp}}			
  \def\Sym{\mathrm{Sym}}				
  \def\U{\mathrm{U}}					

\def\cx{\mathbb C}
\def\rl{\mathbb R}
\def\Z{\mathbb Z}

\def\A{\mathcal A}
\def\C{\mathcal C}
\def\D{\mathcal D}
\def\E{\mathcal E}
\def\G{\mathcal G}
\def\K{\overline{\mathcal K}}
\def\L{\mathcal L}
\def\M{\mathcal M}
\def\P{\mathcal P}

\def\Aut{\mathrm{Aut}}

\def\hol{\mathrm{hol}}
\def\Im{\mathrm{Im}}
\def\Lie{\mathrm{Lie}}

\def\img{\mathrm{image}}
\def\Re{\mathrm{Re}}

\def\thetitle{The Berry connection of the Ginzburg--Landau vortices}
\def\theauthors{\'Akos Nagy}
\def\thekeywords{Berry connection, Ginzburg--Landau vortices}

\title{\thetitle}
\author{\theauthors}
\date{\today}

\address[\'Akos Nagy]{Department of Mathematics, Michigan State University, East Lansing, MI 48824, USA}
\email{\href{mailto:contact@akosnagy.com}{contact@akosnagy.com}}
\urladdr{\href{http://akosnagy.com/}{akosnagy.com}}

\calclayout
\hypersetup{pdffitwindow=true, unicode=false, pdffitwindow=true, pdftoolbar=false, pdfmenubar=false, pdfstartview={FitH}, pdftitle={\thetitle}, pdfkeywords={\thekeywords}, pdfauthor={\theauthors}, colorlinks=true, linkcolor=black, citecolor=black, filecolor=black, urlcolor=blue}

\begin{document}

\begin{abstract}
We analyze 2-dimensional Ginzburg--Landau vortices at critical coupling, and establish asymptotic formulas for the tangent vectors of the vortex moduli space using theorems of Taubes and Bradlow.  We then compute the corresponding Berry curvature and holonomy in the large area limit.
\end{abstract}

\maketitle

\section*{Introduction}

The Ginzburg--Landau theory is a phenomenological model for superconductivity, introduced in \cite{GL50};  for a more modern review see \cite{AK02}.  The theory gives variational equations --- the Ginzburg--Landau equations --- for an Abelian gauge field and a complex scalar field.  The gauge field is the EM vector potential, while the norm of the scalar field is the order parameter of the superconducting phase.  The order parameter can be interpreted as the wave function of the so-called BCS ground state, a single quantum state occupied by a large number of Cooper pairs.

This paper focuses on certain static solutions of the 2-dimensional Ginzburg--Landau equations called $\tau$-vortices.  Physicists regard the number $\tau$ as a coupling constant, sometimes called the vortex-size.  Mathematically, $\tau$ is a scaling parameter for the metric.  The geometry of $\tau$-vortices has been studied since \cites{JT80,B90}, and there is a large literature on the subject;  cf.  \cites{MN99,MS04,CM05,B06,B11,DDM13,BR14,MM15}.

\bigskip

Families of operators in quantum physics carry canonical connections.  This idea was introduced by Berry in \cite{B84}, and generalized by Aharonov and Anandan in \cite{AA87}.  These so-called Berry connections were used, for example, to understand the Quantum Hall Effect \cite{K85}.

In gauge theories --- including the Ginzburg--Landau theory --- the Berry connection can  be understood geometrically as follows: the space $\P$ of solutions of gauge invariant equations is an infinite dimensional principal bundle over the part of moduli space $\M$ where the action of the gauge group $\G$ is free.  Thus if all solutions are irreducible, then $\P$ is an infinite dimensional principal $\G$-bundle
$$
\xymatrix{\P \ar[d]^\G \\
\M}
$$
over $\M$.  The canonical $L^2$-metric of $\P$ defines a horizontal distribution --- the orthogonal complement of the gauge directions --- which defines the Berry connection.

A connection on a principal $G$-bundle $P \rightarrow X$ defines parallel transport:  for each smooth map $\Gamma : [0,1] \rightarrow X$, parallel transport around $\Gamma$ is a $G$-equivariant isomorphism from the fiber at $\Gamma \left( 0 \right)$ to the fiber at $\Gamma \left( 1 \right)$.  If $\Gamma$ is a closed loop, then $\Gamma \left( 0 \right) = \Gamma \left( 1 \right)$, and the corresponding parallel transport is called holonomy.  If $G$ is Abelian, then the holonomy is given by the action of an element in $G$.

Holonomies of the Berry connection are gauge transformations, which have a physical interpretation:  they describe the adiabatic evolution of the state of the system, which is its behavior under slow changes in the physical parameters such as external fields or coupling constants.

\medskip

This paper investigates the Berry connection of the $\tau$-vortex moduli space associated to a degree $d$ Hermitian line bundle over a closed, oriented Riemannian surface $\Sigma$.  The Berry holonomy assigns a gauge transformation $g_\tau$ to each closed curve $\Gamma$ in $\M_\tau$.  These gauge transformations are $\U \left( 1 \right)$-valued smooth functions on $\Sigma$.  When $d$ is positive and $\tau$ is greater than the geometry-dependent constant $\tau_0 = \tfrac{4 \pi d}{\Ar \left( \Sigma \right)}$, then the moduli space, $\M_\tau$, is identified with the $d$-fold symmetric power of the surface $\Sym^d \left( \Sigma \right)$.  A closed curve $\Gamma$ in $\Sym^d \left( \Sigma \right)$ defines a 1-cycle in $\Sigma$, called $\rm{sh} \left( \Gamma \right)$, the shadow of $\Gamma$, which is constructed by choosing a lift of $\Gamma$ to $\Sigma^d \rightarrow \Sym^d \left( \Sigma \right)$, and taking the union of the non-constant curves appearing in the lift (see \cref{eq:loops} for precise definition).  The main theorem of this paper gives a complete topological and analytical description of these gauge transformations in terms of the shadow:

\begin{main}
\hypertarget{main:hol}{{\rm [The Berry holonomy of the $\tau$-vortex principal bundle]}}  Let $g_\tau \in \G$ be the Berry holonomy of a smooth curve $\Gamma$ in the $\tau$-vortex moduli space $\M_\tau$, and $\rm{sh} \left( \Gamma \right)$ be the closed 1-cycle in $\Sigma$ defined in \cref{eq:loops}.  Then the following properties hold as $\tau \to \infty$:
\begin{enumerate}
 \item{\rm [Convergence]}  $g_\tau \rightarrow 1$ in the $C^1$-topology on compact sets of $\Sigma - \rm{sh} \left( \Gamma \right)$.
 \item{\rm [Crossing]}  Let $j: [0,1] \rightarrow \Sigma$ be a smooth path that intersects $\rm{sh} \left( \Gamma \right)$ transversally and positively once, and write $g_\tau \circ j = \exp \left( 2 \pi i \varphi_\tau \right)$.  Then $\varphi_\tau \left( 1 \right) - \varphi_\tau \left( 0 \right) \rightarrow 1$.
 \item{\rm [Concentration]}  As a 1-current, $\tfrac{1}{2 \pi i} g_\tau^{-1} dg_\tau$ converges to the 1-current defined by $\rm{sh} \left( \Gamma \right)$.
\end{enumerate}

The map $\Gamma \mapsto g_\tau$ induces a pairing
\begin{equation}
\hol_\star: \ H_1 \left( \Sigma; \Z \right) \rightarrow H^1 \left( \Sigma; \Z \right),
\end{equation}
defined in \eqref{eq:hol}.
\begin{enumerate}
 \item[(4)]{\rm [Duality]}  For all $\tau > \tau_0$, the homomorphism $\hol_\star$ is Poincar\'e duality.
\end{enumerate}
\end{main}

When $\Gamma$ is a positively oriented, bounding single vortex loop, or a positively oriented vortex interchange (see \Cref{sec:hol} for precise definitions) our main theorem implies that the corresponding holonomy can be written as $g_\tau = \exp \left( 2 \pi i f_\tau \right)$, for a real function $f_\tau$ on $\Sigma$.  Moreover, $f_\tau$ can be chosen so that it converges to 1 on the inside of the curve, and to 0 on the outside.  This makes physicists' intuition about the holonomy precise;  cf. \cite{I01}.

As mentioned above, $\tau$ is a scaling parameter for the metric: one can look at the Ginzburg--Landau theory with $\tau = 1$ fixed, but the K\"ahler form $\omega$ scaled as $\omega_t = t^2 \omega$.  Our results, including the \hyperlink{main:hol}{Main Theorem} above, can be reinterpreted as statements about the large area limit (i.e. $t \rightarrow \infty$), which can be more directly related to physics (see \Cref{sec:area} for details).

\bigskip

The paper is organized as follows.  In \Cref{sec:glr}, we give a brief introduction to the geometry of the $\tau$-vortex equations on a closed surface, derive the tangent space equations of the $\tau$-vortex moduli space, and then recast them in a compact form.  In \Cref{sec:asy}, we use theorems of Taubes and Bradlow to prove a technical result, \Cref{thm:tan}, which establishes asymptotic formulas for the tangent vectors of the of the $\tau$-vortex moduli space.  In \Cref{sec:ber}, we introduce the Berry connection associated to this problem.  We then prove asymptotic formulas for the Berry curvature in \Cref{sec:cur}.  In \Cref{sec:hol}, we prove our \hyperlink{main:hol}{Main Theorem};  the proofs are applications of \Cref{thm:tan}.  \Cref{sec:area} discusses the large area limit.

\smallskip

\subsection*{Acknowledgment}  I wish to thank my advisor Tom Parker for guiding me to this topic and for his constant help during the preparation of this paper.  I am grateful for the support of Jeff Schenker via the NSF grant DMS-0846325.  I also benefited from the discussions with David Duncan, Manos Maridakis, and Tim Nguyen.  Finally I thank the Referee for the incredibly helpful revision.

\smallskip

\section{Ginzburg--Landau Theory on Closed Surfaces}
\label{sec:glr}

\subsection{The $\mathbf{\tau}$-vortex equations}
\label{sec:vor}

Let $\Sigma$ be a closed surface with K\"ahler form $\omega$, compatible complex structure $J$, and Riemannian metric $\omega \left( - , J (-) \right)$.  Let $L \rightarrow \Sigma$ be a smooth complex line bundle of positive degree $d$ with Hermitian metric $h$.  For each unitary connection $\nabla$, and smooth section $\Phi$, consider the {\em Ginzburg--Landau free energy}:
\begin{equation}
\E_{\lambda, \tau} \left( \nabla, \Phi \right) = \int\limits_\Sigma \left( \left| F_\nabla \right|^2 + \left| \nabla \Phi \right|^2 + \lambda w^2 \right) \omega,  \label{eq:glf}
\end{equation}
where $\lambda, \tau \in \rl_+$ are coupling constants, $F_\nabla$ is the curvature of $\nabla$, and
\begin{equation}
w = \frac{1}{2} \left( \tau - \left| \Phi \right|^2 \right).  \label{eq:wdef}
\end{equation}
The Euler-Lagrange equations of the \hyperref[eq:glf]{energy \eqref{eq:glf}} are the {\em Ginzburg--Landau equations}:
\begin{subequations}
\begin{align}
d^*F_\nabla + i \: \Im \left( h \left( \Phi, \nabla \Phi \right) \right) &= 0  \label{eq:gl1}  \\
\nabla^* \nabla \Phi - \lambda w \Phi  &= 0.  \label{eq:gl2}
\end{align}
\end{subequations}
When $\lambda = 1$, the \hyperref[eq:glf]{energy \eqref{eq:glf}} can be integrated by parts and rewritten as different sum of non-negative terms, and get the lower bound $2 \pi \tau d$. The minimizers satisfy the {\em $\tau$-vortex equations}:
\begin{subequations}
\begin{align}
i \Lambda F_\nabla - w &= 0 \label{eq:vo1} \\
\del_\nabla \Phi &= 0, \label{eq:vo2}
\end{align}
\end{subequations}
where $\Lambda F_\nabla$ is the inner product of the K\"ahler form $\omega$ and the curvature of $\nabla$, and $\del_\nabla = \nabla^{0,1}$ is the Cauchy-Riemann operator corresponding to $\nabla$.  Solutions $\left( \nabla, \Phi \right)$ to the first order \cref{eq:vo1,eq:vo2} automatically satisfy the second order \cref{eq:gl1,eq:gl2}.

\smallskip

\subsection{The $\tau$-vortex moduli space}
\label{sec:mod}

As is standard in gauge theory, we work with the Sobolev $W^{k,p}$-completions of the space of connections and fields.  Let $\C_L$ be the $W^{1,2}$-closure of the affine space of smooth unitary connections on $L$ and $\Omega^0_L$ be the $W^{1,2}$-closure of the vector space of smooth sections of $L$.  Similarly, let $\Omega^k$, and $\Omega^k_L$ be the $W^{1,2}$-closure of $k$-forms, and $L$-valued $k$-forms respectively.  The corresponding gauge group $\G$ is the $W^{2,2}$-closure of $\Aut \left( L \right)$ in the $W^{2,2}$-topology.  The gauge group is canonically isomorphic to the infinite dimensional Abelian Lie group $W^{2,2} \left( \Sigma, \U \left( 1 \right) \right)$, whose Lie algebra is $W^{2,2} \left( \Sigma; i \rl \right)$.  Elements $g \in \Aut \left( L \right)$ act on $\C_L \times \Omega^0_L$ as $g \left( \nabla, \Phi \right) = \left( g \circ \nabla \circ g^{-1}, g \Phi \right)$, and this defines a smooth action of $\G$ on $\C_L \times \Omega^0_L$.  Finally, the \hyperref[eq:glf]{energy \eqref{eq:glf}} extends to a smooth function on $\C_L \times \Omega^0_L$.  The space $\P_\tau$ of absolute minimizers of the extended energy is an infinite dimensional submanifold of $\C_L \times \Omega^0_L$.  Due to the gauge invariance of \hyperref[eq:glf]{energy \eqref{eq:glf}}, $\G$ acts on $\P_\tau$, and every critical point is gauge equivalent to a smooth one by elliptic regularity.  The {\em $\tau$-vortex moduli space} is the quotient space $\M_\tau = \P_\tau / \G$.  Elements of $\P_\tau$ are called {\em $\tau$-vortex fields}, while elements of $\M_\tau$ (gauge equivalence classes of $\tau$-vortex fields) are called {\em $\tau$-vortices}.  For brevity, we sometimes write $\tau$-vortex fields as $\vor = \left( \nabla, \Phi \right) \in \P_\tau$ and the corresponding $\tau$-vortices as $\left[ \vor \right] = \left[ \nabla, \Phi \right] \in \M_\tau$.

\smallskip

There is a geometry-dependent constant $\tau_0 = \tfrac{4 \pi d}{\Ar \left( \Sigma \right)}$, called the {\em Bradlow limit}, with the property that if $\tau < \tau_0$, then the moduli space is empty and if $\tau > \tau_0$, then there is a canonical bijection between $\M_\tau$ and the space of effective, degree $d$ divisors; cf. \cite{B90}*{Theorem~4.6}.  This space is also canonically diffeomorphic to the $d$-fold symmetric product of the surface $\Sym^d \left( \Sigma \right)$, which is the quotient of the $d$-fold product $\Sigma^{\times d} = \Sigma \times \ldots \times \Sigma$ by the action of the permutation group $S_d$.  Although this action is not free, the quotient is a smooth K\"ahler manifold of real dimension $2d$.  For each value of $\tau > \tau_0$, there is a canonical $L^2$-K\"ahler structure (see \Cref{sec:horiz}).  In the borderline $\tau = \tau_0$ case, the $\Phi$-field vanishes everywhere and the moduli space is in one-to-one correspondence with the moduli space holomorphic line bundles of degree $d$ \cite{B90}*{Theorem~4.7}.  Accordingly, we focus on the $\tau > \tau_0$ case in this paper.

When $\tau > \tau_0$, Bradlow's map from $\M_\tau$ to $\Sym^d \left( \Sigma \right)$ is easy to understand:  By integrating \cref{eq:vo1}, one sees that the $L^2$-norm of $\Phi$ is positive.  On the other hand, $\Phi$ is a holomorphic section of $L$ by \cref{eq:vo2}.  Since $\Phi$ is a non-vanishing holomorphic section it defines an effective, degree $d$ divisor, giving us the desired map.  The inverse of this map is much harder to understand an involves non-linear elliptic theory.

\smallskip

Much of this picture carries over to open surfaces, even with infinite area (for example $\Sigma = \cx$), if one imposes proper integrability conditions;  cf. \cite{T84a}.  For simplicity we will always assume that $\Sigma$ is compact.  Furthermore the moduli space is empty for $d < 0$, and a single point for $d = 0$.  Thus we will always assume that $d > 0$ in this paper.

\smallskip

\subsection{The horizontal distribution}
\label{sec:horiz}

The tangent space at any point of the affine space $\C_L \times \Omega^0_L$ is the underlying vector space $i \Omega^1 \oplus \Omega^0_L$.  The tangent space of $\P_\tau$ is described in the next lemma.

\begin{lem} \label{lem:tb} The tangent space of $\P_\tau$ at the $\tau$-vortex field $\vor = \left( \nabla, \Phi \right)$ is the vector space of pairs $\left( a, \psi \right) \in i \Omega^1 \oplus \Omega^0_L$ that satisfy
\begin{subequations}
\begin{align}
i \Lambda da + \Re \left( h \left( \psi, \Phi \right) \right) &= 0 \label{eq:tb1} \\
\del_\nabla \psi + a^{0,1} \Phi &= 0. \label{eq:tb2}
\end{align}
\end{subequations}
\end{lem}

\begin{proof}
The linearization of \cref{eq:vo1,eq:vo2} in the direction of $\left( a, \psi \right)$ is:
\begin{align}
\lim\limits_{t \rightarrow 0} \frac{1}{t} \left( i \Lambda F_{\nabla + t \: a} - \frac{1}{2} \left( \tau - \left| \Phi + t \: \psi \right|^2 \right) \right) &= i \Lambda da + \Re \left( h \left( \psi, \Phi \right) \right) \\
\lim\limits_{t \rightarrow 0} \frac{1}{t} \left( \del_{\nabla + t \: a} \left( \Phi + t \: \psi \right) \right) &= \del_\nabla \psi + a^{0,1} \Phi,
\end{align}
where we used that $\vor$ is a $\tau$-vortex field.  This completes the proof, since the tangent space is the kernel of the linearization of \cref{eq:vo1,eq:vo2}.  \end{proof}

The affine space $\C_L \times \Omega^0_L$ has a canonical $L^2$-metric given by
\begin{equation}
\left\langle \left( a, \psi \right) \middle| \left( a', \psi' \right) \right\rangle = \int\limits_\Sigma \left( a \wedge \overline{} a' + \Re \left( h \left( \psi , \psi' \right) \right) \omega \right) \label{eq:riem}
\end{equation}
where $\overline{}$ is the (conjugate-linear) Hodge operator of the Riemannian metric of $\Sigma$.  One can check that the restriction of the \hyperref[eq:riem]{$L^2$-metric~\eqref{eq:riem}} to the solutions of \cref{eq:tb1,eq:tb2} makes $\P_\tau$ a smooth, weak Riemannian manifold, and that gauge transformations act isometrically on $\P_\tau$.

The pushforward of the tangent space $T_1 \G$ by the gauge action is called the {\em vertical subspace} of $T_\vor \P_\tau$.  We define the {\em horizontal subspace} of $T_\vor \P_\tau$ to be the orthogonal complement of the vertical subspace by the \hyperref[eq:riem]{$L^2$-metric~\eqref{eq:riem}}.  Since $\M_\tau = \P_\tau / \G$, the horizontal subspace is canonically isomorphic to the tangent space $T_{\left[ \vor \right]} \M_\tau$ of the moduli space.  The next lemma shows that the horizontal subspace is also the kernel of a first order linear elliptic operator. 

\begin{lem} \label{lem:tm}  The horizontal subspace of $T_\vor \P_\tau$, at the $\tau$-vortex field $\vor = \left( \nabla, \Phi \right) \in \P_\tau$, is the vector space of pairs $\left( a, \psi \right) \in i \Omega^1 \oplus \Omega^0_L$ that satisfy
\begin{subequations}
\begin{align}
\left( i \Lambda d + d^* \right) a + h \left( \psi, \Phi \right) &= 0 \label{eq:tm1} \\
\del_\nabla \psi + a^{0,1} \Phi &= 0.  \label{eq:tm2}
\end{align}
\end{subequations}
\end{lem}

\begin{proof}
The real part of \cref{eq:tm1} is \cref{eq:tb1} and \cref{eq:tm2} is \cref{eq:tb2}; thus solutions of \cref{eq:tm1,eq:tm2} are in $T_\vor \P_\tau$.  To finish the proof, we must check that a pair $\left( a, \psi \right)$ in $T_\vor \P_\tau$ is orthogonal to the vertical subspace at $\vor$ exactly if \cref{eq:tm1,eq:tm2} hold.  The pushforward of $if \in \Lie \left( \G \right)$ at $\vor$ is given by $X_f \left( \vor \right) = \left( - i df, i f \Phi \right)$,  hence horizontal vectors are pairs $\left( a, \psi \right)$, that satisfy the following equation for every $f \in C^\infty \left( \Sigma; \rl \right)$:
\begin{equation}
0 = \left\langle \left( a, \psi \right) \middle| \left( - i df, if \Phi \right) \right\rangle = \int\limits_\Sigma \left( a \wedge \overline{} \left(- i df \right) + \Re \left( h \left( \psi, if \Phi \right) \right) \omega \right).
\end{equation}
Integrating the right-hand side by parts yields
\begin{equation}
0 = \int\limits_\Sigma \left(d^* a + i \: \Im \left( h \left( \psi, \Phi \right) \right) \right) i f \omega.
\end{equation}
Because this holds for all $f$, we conclude that $\left( a, \psi \right)$ is orthogonal to the vertical subspace at $\vor$ exactly if $d^* a + i \: \Im \left( h \left( \psi, \Phi \right) \right) = 0$ holds.  Adding this (purely imaginary) equation to the (purely real) \cref{eq:tb1} gives \cref{eq:tm1}.  \end{proof}

\smallskip

\Cref{eq:tm1,eq:tm2} depend on the choice of $\vor$, but if $\left( a, \psi \right)$ is a solution of \cref{eq:tm1,eq:tb2} for $\vor$ and $g \in \G$, then $(a, g \psi)$ is a solution of \cref{eq:tm1,eq:tb2} for $g \left( \vor \right) = \left( \nabla + g dg^{-1}, g \Phi \right)$.  Since gauge transformations act isometrically, the $L^2$-metric on the horizontal subspaces of $T \P_\tau$ descends to a Riemannian metric on $\M_\tau$.

Let $\overline{K}$ be the anti-canonical bundle of $\Sigma$, and $\Omega^{0,1} = \Omega^0_{\overline{K}}$, the $W^{1,2}$-completion of the space of smooth sections of $\overline{K}$.  We recast \cref{eq:tm1,eq:tm2} in a more geometric way in the next lemma.

\begin{lem} \label{lem:newtm}
\Cref{eq:tm1,eq:tm2} are equivalent to the following pair of equations on $\left( \alpha, \psi \right) \in \Omega^{0,1} \oplus \Omega^0_L$:
\begin{subequations}
\begin{align}
\sqrt{2} \: \del^* \alpha - h \left( \Phi, \psi \right) &= 0 \label{eq:newtm1} \\
\sqrt{2} \: \del_\nabla \psi + \alpha \Phi &= 0. \label{eq:newtm2}
\end{align}
\end{subequations}
Moreover, the unitary bundle isomorphism
\begin{equation}
\left( a, \psi \right) \mapsto \left( \frac{1}{\sqrt{2}} \left( a + i \overline{} a \right), \: \psi \right)  \label{eq:ua}
\end{equation}
interchanges solutions of \cref{eq:tm1,eq:tm2} and solutions of \cref{eq:newtm1,eq:newtm2}.  \end{lem}

\begin{proof}
A complex 1-form $\alpha$ is in $\Omega^{0,1}$ exactly if $\overline{\alpha} = i \overline{} \alpha$.  For $a \in i \Omega^1$, define the unitary map $u$ by
\begin{equation}
u(a) = \frac{1}{\sqrt{2}} \left( a + i \overline{} a \right). \label{eq:ua2}
\end{equation}
Using $\overline{}^2 a = - a = \overline{a}$, we see that $\overline{} u(a) = - i \overline{u(a)}$, and thus $u(a) \in \Omega^{0,1}$.  Set $\alpha = u(a)$.  With this notation $a^{0,1} = \tfrac{\alpha}{\sqrt{2}}$, which proves the equivalence of \cref{eq:tm2,eq:newtm2}.  The K\"ahler identities yield $\left( i \Lambda d + d^* \right) a = - \sqrt{2} \: \partial^* \overline{\alpha}$, which is equivalent to \cref{eq:newtm1}.  \end{proof}

\smallskip

The vector space of solutions to \cref{eq:newtm1,eq:newtm2} has a canonical almost complex structure coming from the complex structures of $\overline{K}$ and $L$, and this defines an almost complex structure for $\M_\tau$.  Mundet i Riera \cite{R00} showed that this structure is integrable, and together with the $L^2$-metric it makes $\M_\tau$ a K\"ahler manifold.

\medskip

To put \cref{eq:newtm1,eq:newtm2} in a more compact form, note that they are equivalent to the single equation
\begin{equation}
\L_\vor \left(a, \psi \right) = 0,
\end{equation}
where $\L_\vor = \D_\nabla + \A_\Phi$ is defined as
\begin{align}
\D_\nabla &: \Omega^{0,1} \oplus \Omega^0_L \rightarrow \Omega^0 \oplus \Omega^{0,1}_L; \quad \left( \alpha, \psi \right) \mapsto \left( \sqrt{2} \: \del^* \alpha, \: \sqrt{2} \: \del_\nabla \psi \right) \\
\A_\Phi &: \Omega^{0,1} \oplus \Omega^0_L \rightarrow \Omega^0 \oplus \Omega^{0,1}_L; \quad \left( \alpha, \psi \right) \mapsto \left( - h( \Phi, \psi), \: \alpha \Phi \right),
\end{align}
The operator $\D_\nabla$ is a first order elliptic differential operator, and the operator $\A_\Phi$ is a bundle map.  Straightforward computation shows that for all $Z \in \Omega^{0,1} \oplus \Omega^0_L$
\begin{equation}
\A_\Phi \A_\Phi^* \left( Z \right) = \left| \Phi \right|^2 Z. \label{eq:asq}
\end{equation}
Thus $\A_\Phi^*$ is non-degenerate on the complement of the divisor of $\Phi$.  Note that $\D_\nabla$, $\A_\Phi$, and hence $\L_{\left( \nabla, \Phi \right)}$ make sense for any pair $\left( \nabla, \Phi \right) \in \C_L \times \Omega^0_L$.

\begin{lem} \label{lem:van} Let $\vor = \left( \nabla, \Phi \right) \in \C_L \times \Omega^0_L$ be a pair such that $\del_\nabla \Phi=0$.  Then the operator $\D_\nabla \A_\Phi^* + \A_\Phi \D_\nabla^*$ is identically zero.
\end{lem}

\begin{proof}
The adjoint operators are
\begin{subequations}
\begin{align}
\D_\nabla^* \left( f, \xi \right) &= \left( \sqrt{2} \: \mathit{\del f}, \: \sqrt{2} \: \del^*_\nabla \xi \right) \label{eq:dnabla} \\
\A_\Phi^* \left( f, \xi \right) &= \left( h \left( \Phi, \xi \right), - f \Phi \right) \label{eq:aphi}
\end{align}
\end{subequations}
for any $\left( f, \xi \right) \in \Omega^0 \oplus \Omega^{0,1}_L$.  The lemma follows from \cref{eq:dnabla,eq:aphi}, the holomorphicity of $\Phi$, and the definitions of $\D_\nabla$ and $\A_\Phi$.  \end{proof}

\begin{cor} \label{cor:ind}
Let $\vor = \left( \nabla, \Phi \right) \in \C_L \times \Omega^0_L$ be a pair such that $\del_\nabla \Phi=0$.  Then $\ker \left( \L_\vor^* \right)$ is trivial.
\end{cor}

\begin{proof}
From \Cref{lem:van} and \cref{eq:asq}, we obtain $\L_\vor \L_\vor^* = \D_\nabla \D_\nabla^* + \A_\Phi \A_\Phi^*$. Hence if $Z$ is in the kernel of $\L_\vor^*$, then
\begin{equation}
0 = \| \L_\vor^* \left( Z \right) \|^2_{L^2} = \| \D_\nabla^* \left( Z \right) \|^2_{L^2} + \| \A_\Phi^* \left( Z \right) \|^2_{L^2},
\end{equation}
which implies that both terms on the right vanish.  By \cref{eq:asq}, $Z$ vanishes where $\Phi$ does not, which is the complement of a finite set.  But then $Z$ vanishes everywhere by continuity.  Hence the kernel of $\L_\vor^*$ is trivial.  \end{proof}

\smallskip

\section{The Asymptotic form of Horizontal Vectors}
\label{sec:asy}

In this section we will use of the following results of Taubes and Bradlow about the large $\tau$ behavior of $\tau$-vortex fields.  Recall from \cref{eq:wdef} that
$$
w = \frac{1}{2} \left( \tau - |\Phi|^2 \right).
$$

\begin{thm} {\rm [Bradlow and Taubes]} \label{thm:bnt}
There is a positive number $c = c \left( \Sigma, \omega, J, L, h \right)$ such that each $\tau$-vortex field $\vor = \left( \nabla, \Phi \right) \in \P_\tau$ satisfies
\begin{subequations}
\begin{align}
\left| \Phi \right|^2 &\leqslant \tau  \label[ineq]{ineq:bnt1} \\
w + \left| \nabla \Phi \right| &\leqslant c \tau \exp \left( - \frac{\sqrt{\tau} \dist_D}{c} \right)  \label[ineq]{ineq:bnt2},
\end{align}
\end{subequations}
where $\dist_D$ is the distance from the divisor $D = \Phi^{-1} \left( 0 \right)$, and $w$ is defined in \cref{eq:wdef}.
\end{thm}

\begin{proof}
In \cite{B90}*{Proposition~5.2} Bradlow showed \cref{ineq:bnt1}, using the fact that $\tau$-vortex fields satisfy the elliptic equation
\begin{equation}
\left( \Delta + \left| \Phi \right|^2 \right) w = \left| \partial_\nabla \Phi \right|^2. \label{eq:bra}
\end{equation}
The right-hand side is positive away from a finite set, so the maximum principle and \cref{eq:wdef} implies \cref{ineq:bnt1}.  \Cref{ineq:bnt2} was proved in \cite{T99}*{Lemma~3.3}.
\end{proof}

We call a divisor {\em simple} if the multiplicity of every divisor point is 1.

\begin{lem} \label{lem:delta}
Fix a simple divisor $D \in \Sym^d \left( \Sigma \right)$ and a corresponding $\tau$-vortex field $\vor = \left( \nabla, \Phi \right)$ with $\tau > \tau_0$.  The smooth function
\begin{equation}
h_{D, \tau} = \frac{1}{2 \pi \tau} \left( |\partial_\nabla \Phi|^2 + 2 w^2 \right) \label{eq:hdtau}
\end{equation}
depends only on $D$ and $\tau$, but not on the choice of $\vor$.  Moreover,
\begin{equation}
\lim\limits_{\tau \rightarrow \infty} h_{D, \tau} = \sum\limits_{{p \in D}} \delta_p, \label{eq:limhdtau}
\end{equation}
in the sense of measures, where $\delta_p$ is the Dirac measure concentrated at the point $p \in \Sigma$.
\end{lem}

\begin{proof}
Every term in \cref{eq:hdtau} is gauge invariant, which proves the independence of the choice of $\vor$ for $D$.  Using \cref{eq:bra} we get $h_{D, \tau} = \tfrac{1}{2 \pi \tau} \left( \Delta w + \tau w \right)$, hence for any smooth function $f$:
\begin{align}
\int\limits_\Sigma h_{D, \tau} f \omega &= \frac{1}{2 \pi \tau} \int\limits_\Sigma \left( \Delta w + \tau w \right) f \omega \\
&= \frac{1}{2 \pi \tau} \int\limits_\Sigma w \left( \Delta f \right) \omega + \frac{1}{2 \pi} \int\limits_\Sigma w f \omega.
\end{align}
By \cite{HJS96}*{Theorem~1.1}, $w$ converges to $2 \pi \delta_D$ in the sense of measures as $\tau \rightarrow \infty$.  Thus the first term converges to 0, and the second term converges to $\sum\limits_{p \in D} f \left( p \right)$, which completes the proof.
\end{proof}

The space of simple divisors, $\Sym^d_s \left( \Sigma \right)$, is an open dense set in $\Sym^d \left( \Sigma \right)$, and its complement is called the {\em big diagonal}.  When $D$ is simple, a tangent vector in $T_D \Sym^d \left( \Sigma \right)$ can be given by specifying a tangent vector to $\Sigma$ at each divisor point. Thus the rank $d$ complex vector bundle $\K \rightarrow \Sym^d_s \left( \Sigma \right)$ defined by
\begin{equation}
\K_D = \mathop{\oplus}\limits_{p \in D} \overline{K}_p
\end{equation}
is isomorphic to $T^{0,1} \Sym^d_s \left( \Sigma \right)$.  We next use ideas of \cite{T99}*{Lemma~3.3} to construct an almost unitary isomorphism from $\K$ to $T^{0,1} \Sym^d_s \left( \Sigma \right)$.

Fix a simple divisor $D$, and let $\vor$ be a corresponding $\tau$-vortex field.  Define
\begin{equation}
\rho_D = \min \left( \left\{ \dist \left( p, q \right) \middle| p, q \in D \: \& \: p \neq q \right\} \cup \left\{ \mathrm{inj} \left( \Sigma, \omega \right) \right\} \right), \label{eq:rhod}
\end{equation}
where $\mathrm{inj} \left( \Sigma, \omega \right)$ is the injectivity radius of the metric.  Let $\chi$ be a smooth function on $[0, \infty)$ that satisfies $0 \leqslant \chi \leqslant 1$, $\chi|_{[0,1]} = 1$, and $\chi|_{[2, \infty)} = 0$, and set
\begin{equation}
\chi_p = \chi \left( \frac{2 \dist_p}{\rho_D} \right).  \label{eq:chip}
\end{equation}
For each $\Theta = \left\{ \theta_p \right\}_{p \in D} \in \K_D$ let $\widehat{\theta}_p$ be the extension of $\theta_p$ to the open ball of radius $\mathrm{inj} \left( \Sigma, \omega \right)$ centered at $p$ using the exponential map. Define a smooth section $\sigma_{\Theta}$ of $\overline{K}$ supported in neighborhood of $D$ by setting
\begin{equation}
\sigma_{\Theta} = \sum\limits_{p \in D} \chi_p \widehat{\theta}_p \label{eq:sigmatheta}
\end{equation}
and extending by 0 to all of $\Sigma$.  Note that $\sigma_{\Theta}$ satisfies
\begin{equation}
\sigma_{\Theta} \left( p \right) = \theta_p \qquad \&  \qquad \left| \nabla \sigma_{\Theta}\right| = O \left( \dist_D \right) \quad \forall p \in D.  \label{eq:sigmabound}
\end{equation}
Finally, for each such $\vor$ and $\Theta$ define $Y_{\vor, \Theta}$ as 
\begin{equation}
Y_{\vor, \Theta} = \frac{1}{\sqrt{2 \pi \tau}} \left( \sqrt{2} w \sigma_{\Theta}, i \Lambda \left( \sigma_{\Theta} \wedge \partial_\nabla \Phi \right) \right) \in \Omega^{0,1} \oplus \Omega^0_L, \label{eq:apt1}
\end{equation}
where again $w = \tfrac{1}{2} \left( \tau - \left| \Phi \right|^2 \right)$.  Note that $Y_{\vor, \Theta}$ is gauge equivariant, that is, for every $g \in \G$:
\begin{equation}
g_* Y_{\vor, \Theta} = Y_{g \left( \vor \right), \Theta}.  \label{eq:yequiv}
\end{equation}

The following analytic result is the key ingredient needed to compute the asymptotic curvature in \Cref{thm:cur} and holonomies in the \hyperlink{main:hol}{Main Theorem}.

\begin{thm} {\rm [The asymptotic form of horizontal vectors]} \label{thm:tan} For every $\vor \in \P_\tau$ and $\Theta \in \K_D$ as above, there is a unique $Z_{\vor, \Theta} \in \Omega^0 \oplus \Omega^{0,1}_L$ such that
\begin{equation}
X_{\vor, \Theta} = Y_{\vor, \Theta} - \L_\vor^* \left( Z_{\vor, \Theta} \right)  \label{eq:apt2}
\end{equation}
is a horizontal tangent vector at $\vor$.  Moreover, the following asymptotic estimates hold:
\begin{enumerate}
 \item{\rm [$L^2$-estimate]} $\| Y_{\vor, \Theta} \|^2_{L^2 \left( \Sigma \right)} \rightarrow \sum\limits_{p \in D} \left| \theta_p \right|^2$ as $\tau \rightarrow \infty$.  
 \item{\rm [Pointwise bound]} $\left| X_{\vor, \Theta} - Y_{\vor, \Theta} \right| = O \left( \tau^{-1/2} \exp \left(- \tfrac{\sqrt{\tau} \dist_D}{c} \right) \right)$,
where $\dist_D$ is the distance from $D$, and $c$ is the positive number from \Cref{thm:bnt}.
\end{enumerate}
\end{thm}

\Cref{eq:apt2} defines a bundle map from $\K$ to $T^{0,1} \Sym^d_s \left( \Sigma \right)$ by
\begin{equation}
\left( D, \Theta \right) \mapsto \left( D, \Pi_* \left( X_{\vor, \Theta} \right) \right),  \label{eq:ktosym}
\end{equation}
where $\vor$ is any $\tau$-vortex field corresponding to the divisor $D$, and $\Pi$ is the projection from $\P_\tau$ to $\M_\tau \cong \Sym^d \left( \Sigma \right)$.  By \cref{eq:yequiv}, the map \eqref{eq:ktosym} does not depend on the choice of $\vor$.  Furthermore, this map is almost unitary by Statements~(1)~and~(2).  Similar results have only been known for flat metrics \cite{T99}*{Lemma~3.3}.

\begin{proof}[Proof of \Cref{thm:tan}:]
Fix $Y_{\vor, \Theta}$ as in \cref{eq:apt1}.  Since $\L_\vor$ is elliptic, and $\ker \left( \L_\vor^* \right) = \{ 0 \}$, by \Cref{cor:ind}, the operator $\L_\vor \L_\vor^*$ is has a bounded inverse $\left( \L_\vor \L_\vor^* \right)^{-1}$.  Thus the equation
\begin{equation}
\L_\vor \left( Y_{\vor, \Theta} - \L_\vor^* \left( Z_{\vor, \Theta} \right) \right) = 0  \label{eq:zvt}
\end{equation}
has a unique solution  for $Z_{\vor, \Theta}$ given by
\begin{equation}
Z_{\vor, \Theta} = \left( \L_\vor \L_\vor^* \right)^{-1} \left( \L_\vor \left( Y_{\vor, \Theta} \right) \right)  \label{eq:zvorth}
\end{equation}
Consequently $X_{\vor, \Theta}$ in \cref{eq:apt2} is horizontal.

The pointwise norm of $Y_{\vor, \Theta}$ satisfies
\begin{equation}
\left| Y_{\vor, \Theta} \right|^2 = h_{D, \tau} \left| \sigma_{\Theta} \right|^2,
\end{equation}
where $h_{D, \tau}$ is defined in \cref{eq:hdtau}.  Using \cref{eq:limhdtau}, one gets Statement~(1).

In order to prove Statement~(2), we put $Z_{\vor, \Theta} = \left( f_{\vor, \Theta}, \xi_{\vor, \Theta} \right) \in \Omega^0 \oplus \Omega^{0,1}_L$ in \cref{eq:zvt} to obtain the equations
\begin{subequations}
\begin{align}
\left( \frac{1}{2} \Delta + \frac{1}{2} \left| \Phi \right|^2 \right) f_{\vor, \Theta} &= \frac{1}{\sqrt{8 \pi \tau}}w \del^* \sigma_{\Theta}, \label{eq:hf} \\
\left( \del_\nabla \del_\nabla^* + \frac{1}{2} \left| \Phi \right|^2 \right) \xi_{\vor, \Theta} &= \frac{i}{\sqrt{4 \pi \tau}} \Lambda \left( \nabla^{0,1} \sigma_{\Theta} \wedge \partial_\nabla \Phi \right).  \label{eq:hxi}
\end{align}
\end{subequations}

Let $G_{\left[ \Phi \right]}$ be the Green's operator of the non-degenerate elliptic operator $H_{\left[ \Phi \right]} = \tfrac{1}{2} \Delta + \tfrac{1}{2} \left| \Phi \right|^2$ on $\Omega^0$.  Both $H_{\left[ \Phi \right]}$ and $G_{\left[ \Phi \right]}$ depend only on the gauge equivalence class of $\Phi$.  By an abuse of notation $G_{\left[ \Phi \right]}$ will also denote the corresponding Green's function, which is a positive, symmetric function on $\Sigma \times \Sigma$ with a logarithmic singularity along the diagonal.  With this definitions, we can write $f_{\vor, \Theta}$ as
\begin{equation}
f_{\vor, \Theta} = \frac{1}{\sqrt{8 \pi \tau}} \int\limits_\Sigma G_{\left[ \Phi \right]} w \del^* \sigma_{\Theta} \omega. \label{eq:fblue}
\end{equation}
Standard elliptic theory gives the following bounds on the Green's function:
\begin{subequations}
\begin{align}
G_{\left[ \Phi \right]} \left( x, y \right) & \leqslant c \left( 1 + \left| \ln \left( \sqrt{\tau} \dist \left( x, y \right) \right) \right| \right) \exp \left( - \frac{\sqrt{\tau} \dist \left( x, y \right)}{c} \right) \label[ineq]{ineq:green1} \\
\left| dG_{\left[ \Phi \right]} \left( x, y \right) \right| & \leqslant  \frac{c}{\sqrt{\tau} \dist \left( x, y \right)} \exp \left( - \frac{\sqrt{\tau} \dist \left( x, y \right)}{c} \right) \label[ineq]{ineq:green2}
\end{align}
\end{subequations}
for some $c \in \rl_+$ independent of $\tau$ or $D$ (see \cite{T99}*{Equation (6.10)}).  Using \cref{eq:fblue,ineq:green1,ineq:green2}, together with the bound on $w$ in \cref{ineq:bnt2} and on $\left| \nabla \sigma_{\Theta} \right|$ in \cref{eq:sigmabound} we get (after possibly increasing $c$)
\begin{subequations}
\begin{align}
\left| f_{\vor, \Theta} \right| &\leqslant \frac{c}{\tau} \exp \left( - \frac{\sqrt{\tau} \dist_D}{c} \right) \label[ineq]{ineq:f1} \\
\left| \del f_{\vor, \Theta} \right| &\leqslant \frac{c}{\sqrt{\tau}} \exp \left( - \frac{\sqrt{\tau} \dist_D}{c} \right).  \label[ineq]{ineq:f2}
\end{align}
\end{subequations}
Before turning our attention to \cref{eq:hxi}, note that we have the following two scalar identities
\begin{subequations}
\begin{align}
\Delta \left| \xi_{\vor, \Theta} \right|^2 &= 2 \Re \left( \left\langle \xi_{\vor, \Theta} \middle| \nabla^* \nabla \xi_{\vor, \Theta} \right\rangle \right) - 2 \left| \nabla \xi_{\vor, \Theta} \right|^2.  \label{eq:deltaxi1} \\
\Delta \left| \xi_{\vor, \Theta} \right|^2 &= 2 \left| \xi_{\vor, \Theta} \right| \Delta \left| \xi_{\vor, \Theta} \right| - 2 \left| d \left| \xi_{\vor, \Theta} \right| \right|^2,  \label{eq:deltaxi2}
\end{align}
\end{subequations}
Using the ``K\"ahler identity'' on $\Omega^{0,1}_L$ 
\begin{equation}
\nabla^* \nabla = 2 \del_\nabla \del_\nabla^* - i \Lambda F_\nabla = 2 \del_\nabla \del_\nabla^* - \frac{1}{2} \left( \tau - \left| \Phi \right|^2 \right),  \label{eq:}
\end{equation}
in \cref{eq:hxi}, together with the Cauchy-Schwarz inequality in \cref{eq:deltaxi1}, and Kato's inequality (cf. \cite{FU91}*{Equation~(6.20)})
\begin{equation}
\left| d \left| \xi_{\vor, \Theta} \right| \right| \leqslant \left| \nabla \xi_{\vor, \Theta} \right| \label[ineq]{ineq:kato}
\end{equation}
gives us
\begin{equation}
\left( \Delta + \frac{1}{2} \tau \right) \left| \xi_{\vor, \Theta} \right| \leqslant \frac{c}{\sqrt{\tau}} \left| \nabla \sigma_{\Theta} \right| \left|\del_\nabla \Phi \right|.  \label{eq:deltatauxi}
\end{equation}
The bound on $\left| \del_\nabla \Phi \right|$ in \cref{ineq:bnt2} and the bound on $\left| \nabla \sigma_{\Theta} \right|$ in \cref{eq:sigmabound}, together with \cref{eq:deltatauxi}, gives us (again after possibly increasing $c$)
\begin{equation}
\left| \xi_{\vor, \Theta} \right| \leqslant \frac{c}{\tau} \exp \left( - \frac{\sqrt{\tau} \dist_D}{c} \right).  \label[ineq]{ineq:xi1}
\end{equation}
Applying $\del_\nabla^*$ to \cref{eq:hxi} gives an elliptic equation on $\del^*_\nabla \xi_{\vor, \Theta}$.  Similarly to the previous computation we get the following inequality:
\begin{equation}
\left( \Delta + \frac{1}{2} \tau \right) \left| \del^*_\nabla \xi_{\vor, \Theta} \right| \leqslant c \sqrt{\tau} \exp \left( - \frac{\sqrt{\tau} \dist_D}{c} \right).
\end{equation}
Thus (after possibly increasing $c$ one last time)
\begin{equation}
\left| \del^*_\nabla \xi_{\vor, \Theta} \right| \leqslant \frac{c}{\sqrt{\tau}} \exp \left( - \frac{\sqrt{\tau} \dist_D}{c} \right).  \label[ineq]{ineq:xi2}
\end{equation}
Finally, \cref{ineq:f1,ineq:f2,ineq:xi1,ineq:xi2} give us
\begin{align}
\left| X_{\vor, \Theta} - Y_{\vor, \Theta} \right| &= \left| \L^*_\vor \left( Z_{\vor, \Theta} \right) \right| \\
&\leqslant \sqrt{2} \left| \del f_{\vor, \Theta} \right| + \left| \Phi \right| \left| f \right| + \sqrt{2} \left| \del^*_\nabla \xi \right| + \left| \Phi \right| \left| \xi \right| \\
&= O \left( \tau^{-1/2} \exp \left( - \frac{\sqrt{\tau} \dist_D}{c} \right) \right),
\end{align}
which completes the proof of Statement~(2).
\end{proof}

\smallskip

\section{The Berry Connection}
\label{sec:ber}

The {\em $\tau$-vortex principal bundle} is the principal $\G$-bundle $\Pi: \P_\tau \rightarrow \M_\tau$ described in \Cref{sec:glr}, with $\Pi \left( \vor \right) = \left[ \vor \right]$.  In \Cref{lem:tb} we constructed a horizontal distribution on the $\tau$-vortex principal bundle, which is the orthogonal complement of the kernel of $\Pi_*$.  This distribution is $\G$-invariant, so is a connection in the distributional sense (cf.  \cite{KN69}*{Chapter~II}), which we call the {\em Berry connection}.  The corresponding connection 1-form is the unique $\Lie \left( \G \right)$-valued 1-form $A$ that satisfies the three conditions:
\begin{enumerate}
 \item $\ker \left( A_\vor \right)$ is the horizontal subspace at $\vor \in \P_\tau$,
 \item $\left( g^*A \right)_{g \left( \vor \right)} = \rm{ad}_g \left( A_\vor \right)$, for all $g \in \G$,
 \item $A \left( X_f \right) = i f$, for all $i f \in W^{2,2} \left( \Sigma; i \rl \right) \cong \Lie \left( \G \right)$, where $X_f \left( \vor \right) = \left( -i df, i f \Phi \right)$, as defined in \Cref{lem:tm}.
\end{enumerate}
The next lemma gives a formula for $A_\vor$.  Recall that for each $\tau$-vortex field $\vor = \left( \nabla, \Phi \right)$, the Green's operator $G_{\left[ \Phi \right]}$ is the inverse of the non-degenerate elliptic operator $H_{\left[ \Phi \right]} = \tfrac{1}{2} \Delta + \tfrac{1}{2} \left| \Phi \right|^2$.

\begin{lem} \label{lem:ab}
The $\Lie(\G)$-valued 1-form on $\P_\tau$ defined as
\begin{equation}
A_\vor \left( a, \psi \right) = - \frac{1}{2} G_{\left[ \Phi \right]} \left( d^* a + i \: \Im(h(\psi, \Phi)) \right) \label{eq:ab}
\end{equation}
is the connection 1-form corresponding to the Berry connection.
\end{lem}

\begin{proof}  The right-hand side of \cref{eq:ab} is the composition of the non-degenerate Green's operator and a $\Lie \left( \G \right)$-valued 1-form.  In the proof of \Cref{lem:tb} we saw that the kernel of this 1-form is exactly the horizontal subspace.  This proves Condition (1) above.

Because $\G$ is Abelian, the adjoint representation of $\G$ is trivial, and hence the Condition (2) reduces to $\left( g^*A \right)_{g \left( \vor \right)} = A_\vor$.  Since $g_* \left( a, \psi \right) = \left( a, g \psi \right)$, we have
\begin{align}
g^* A_{g \left( \vor \right)} \left( a, \psi \right) &= \frac{-1}{2} G_{\left[ \Phi \right]} \left( d^*a + i \: \Im \left( h \left( g \psi, g \Phi \right) \right) \right) \\
&= \frac{-1}{2} G_{\left[ \Phi \right]} \left( d^*a + i \: \Im \left( h \left( \psi, \Phi \right) \right) \right),
\end{align}
thus $g^* A_{g \left( \vor \right)} = A_{\vor}$.  This proves Condition (2).

Finally, we show that $A$ is the canonical isomorphism between the fibers of the vertical bundle and the Lie algebra of $\G$, that is $A \left( X_f \right) = i f$  for every $f \in C^\infty \left( \Sigma; \rl \right)$:
\begin{align}
A_\vor \left( X_f \right) &= \frac{-1}{2} G_{\left[ \Phi \right]} \left( d^* \left( -i df \right) + i \: \Im \left( h \left( i f \Phi, \Phi \right) \right) \right) \omega \\
&= i G_{\left[ \Phi \right]} \left( \left( \frac{1}{2} \Delta f + \frac{1}{2} \left| \Phi \right|^2 \right) f \right) \\
&= if,
\end{align}
thus Condition (3) holds.  \end{proof}

We can use \Cref{lem:ab} to compute the curvature 2-form of the Berry connection.  Since $\G$ is Abelian, the curvature, called the {\em Berry curvature}, is a $\Lie \left( \G  \right)$-valued 2-form which descends to the base space $\M_\tau$.

\begin{thm} \label{thm:ber}
The curvature 2-form of the Berry connection at $\left[ \vor \right] \in \M_\tau$ is
\begin{equation}
\Omega_{\left[ \vor \right]} \left( X, Y \right) = G_{\left[ \Phi \right]} \left( i \: \Im \left( h \left( \psi_X, \psi_Y \right) \right) \right) \label{eq:cur}
\end{equation}
where $\left( a_X, \psi_X \right)$ and $\left( a_Y, \psi_Y \right)$ are the horizontal lifts of $X$ and $Y$, respectively, at $\vor$.  Moreover, \cref{eq:cur} does not depend on the choice of the $\tau$-vortex field $\vor$ representing $\left[ \vor \right]$.
\end{thm}

\begin{proof}
The claim about the independence of the choice $\vor$ is immediate since everything on the right-hand side is gauge invariant.

The curvature is the unique $\Lie \left( \G \right)$-valued 2-form $\Omega$ on $\M_\tau$ that satisfies $\Pi^* \left( \Omega \right) = dA$, where $\Pi$ is the projection from $\P_\tau$ to $\M_\tau$.  Thus it is enough to compute $dA_\vor \left( \left( a_X, \psi_X \right), \left( a_Y, \psi_Y \right) \right)$ and compare it with \cref{eq:cur}.  Recall, that the formula for the exterior derivative
\begin{equation}
dA \left( \widetilde{X}, \widetilde{X} \right) = \widetilde{X} \left( A \left( \widetilde{Y} \right) \right) - \widetilde{Y} \left( A \left( \widetilde{X} \right) \right) - A \left( \left[ \widetilde{X}, \widetilde{Y} \right] \right), \label{eq:ext}
\end{equation}
where $\widetilde{X}$ and $\widetilde{Y}$ are smooth local extensions of $\left( a_X, \psi_X \right)$ and $\left( a_Y, \psi_Y \right)$ respectively.  Choose the extensions so that their Lie bracket vanishes at $\vor$.  Let $\Upsilon_t$ be the local flow generated by $\widetilde{X}$, so $\Upsilon_t \left( \vor \right) = \vor + t \: \left( a_X, \psi_X \right) + O \left( t^2 \right)$.  Since $A \left( \widetilde{Y} \right) = 0$ at $\Upsilon_0 (\vor) = \vor$, we have
\begin{equation}
\widetilde{X}_\vor \left( A \left( \widetilde{Y} \right) \right) = \lim\limits_{t \rightarrow 0} \frac{1}{t} A_{\Upsilon_t \left( \vor \right)} \left( \widetilde{Y} \left( \Upsilon_t \left( \vor \right) \right) \right).
\end{equation}
Note that $\widetilde{Y} \left( \Upsilon_t \left( \vor \right) \right) = \left( \left( \Upsilon_t \right)_* \left( a_Y, \psi_Y \right) \right) + O \left( t^2 \right)$, because $\left[ \widetilde{X}, \widetilde{Y} \right]_\vor = 0$.  Finally, let us write
\begin{equation}
G_{\left[ \Phi + t \: \psi_X + O \left( t^2 \right) \right]} = G_{\left[ \Phi \right]} + t \: G_{\left[ \Phi \right]}^X + O \left( t^2 \right).
\end{equation}
Keeping only the linear terms, we obtain
\begin{align}
\widetilde{X}_\vor \left( A \left( \widetilde{Y} \right) \right) &= \lim\limits_{t \rightarrow 0} \frac{1}{t} A_{\Upsilon_t \left( \vor \right)} \left( \widetilde{Y} \left( \Upsilon_t \left( \vor \right) \right) \right) \\
&= \lim\limits_{t \rightarrow 0} - \frac{1}{2t} G_{\left[ \Phi + t \: \psi_X + O \left( t^2 \right) \right]} \left( d^*a_Y + i \: \Im \left( h \left( \psi_Y, \Phi + t \: \psi_X \right) \right) + O \left( t^2 \right) \right) \\
&= \lim\limits_{t \rightarrow 0} - \frac{1}{2t} \left( it G_{\left[ \Phi \right]} \left( \Im \left( h \left( \psi_Y, \psi_X \right) \right) \right) + t \: G_{\left[ \Phi \right]}^X \left( d^*a_Y + i \: \Im \left( h \left( \psi_Y, \Phi \right) \right) \right) \right) \\
&= - \frac{i}{2} G_{\left[ \Phi \right]} \left( \Im \left( h \left( \psi_Y, \psi_X \right) \right) \right),
\end{align}
where we used the fact that $d^*a_Y + i \: \Im \left( h \left( \psi_Y, \Phi \right) \right) = 0$ for tangent vectors.  Interchanging $\widetilde{X}$ and $\widetilde{Y}$ changes sign, since $\Im \left( h \left( \psi_Y, \psi_X \right) \right)$ is skew.  Substituting these into \cref{eq:ext}, and noting that the commutator vanishes, gives \cref{eq:cur}.  \end{proof}

\smallskip

\section{The Asymptotic Berry Curvature}
\label{sec:cur}

In this section we use \Cref{thm:tan,thm:ber} to analyze the Berry curvature in the large $\tau$ limit.

As before, let $D$ be a simple divisor, and $\vor = \left( \nabla, \Phi \right)$ be th corresponding $\tau$-vortex field.  For each $p \in D$, choose $\Theta_p = \{ \theta_{p,q} \}_{q \in D} \in \K_D$, so that $|\theta_{p,q}| = \delta_{p,q}$.  Let $\sigma_p = \sigma_{\Theta_p}$ be the corresponding section defined by \cref{eq:sigmatheta}, and let $X_{\vor, \Theta_p} = \left( a_p, \psi_p \right)$, as defined in \Cref{thm:tan}.  By \cref{eq:apt2},
\begin{equation}
X_{\vor, \Theta_p} = Y_{\vor, \Theta_p} - \L_\vor^* \left( Z_{\vor, \Theta_p} \right).  \label{eq:apt3}
\end{equation}
where $Z_{\vor, \Theta_p} = \left( f_p, \xi_p \right) \in \Omega^0 \oplus \Omega^{0,1}_L$.  It is easy to see that in Statement~(2) of \Cref{thm:tan} we can now replace $\dist_D$ with $\dist_p$, the distance from the single point $p$.

The set $\{ X_p \}_{p \in D}$, where $X_p = \Pi_* \left( X_{\vor, \Theta_p} \right) \in T_{\left[ \vor \right]} \M_\tau$, is an asymptotically orthonormal basis for the horizontal subspace at $\vor$, in the sense that as $\tau \rightarrow \infty$
\begin{equation}
\left\langle X_p \middle| X_q  \right\rangle = \delta_{p,q} + O \left( \exp \left(- \frac{\sqrt{\tau} \rho_D}{c} \right) \right) \rightarrow \delta_{p,q}.
\end{equation}
Finally, for each tangent vector $X$, let $X^\flat = \left\langle X \middle| - \right\rangle$ be the metric-dual covector.

\begin{thm} {\rm [The asymptotic Berry curvature]} \label{thm:cur}
There is a positive number $c = c \left( \Sigma, \omega, J, L, h \right)$ such that if $\tau~>~\tau_0~=~\tfrac{4 \pi d}{\Ar \left( \Sigma \right)}$ and $\left[ \vor \right]$ is a simple $\tau$-vortex, then the Berry curvature satisfies
\begin{align}
\Omega_{\left[ \vor \right]} &= \sum\limits_{p, q \in D} \left( \left( \chi_p  \delta_{p,q} \frac{i w}{\pi \tau} + i A_\tau^{p,q} \right) X^\flat_p \wedge \left( i X_q \right)^\flat + i B_\tau^{p,q}  X^\flat_p \wedge X^\flat_q + i C_\tau^{p,q} \left( i X_p \right)^\flat \wedge \left( i X_q \right)^\flat \right),  \label{eq:cb}
\end{align}
where $\chi_p$ as defined in \cref{eq:chip}, and $A^{p,q}_\tau, B^{p,q}_\tau$, and $C^{p,q}_\tau$ are real functions, with
\begin{equation}
\left| A^{p,q}_\tau \right| + \left| B^{p,q}_\tau \right| + \left| C^{p,q}_\tau \right| = O \left( \tau^{-1} \exp \left( - \frac{\sqrt{\tau} \rho_D}{c} \right) \right). \label{eq:coeffbound}
\end{equation}
\end{thm}

\begin{proof}
For $p \neq q$, \Cref{thm:bnt,thm:tan} imply that
\begin{equation}
\left| h \left( \psi_p, \psi_q \right) \right| = |\psi_p| |\psi_q| \leqslant c \exp \left( - \frac{\sqrt{\tau} \left( \dist_p + \dist_q \right)}{c} \right) \leqslant c \exp \left( - \frac{\sqrt{\tau} \rho_D}{c} \right).
\end{equation}
This inequality together with the fact that $G_{\left[ \Phi \right]} \left( 1 \right) = O \left( \tau^{-1} \right)$ from \cref{ineq:green1}, gives \cref{eq:coeffbound} in this case.

In general, for every $p \in D$, \Cref{thm:ber} shows that
\begin{equation}
\frac{1}{i} \Omega_{\left[ \vor \right]} \left( X_p, i X_p \right) = G_{\left[ \Phi \right]} \left( \left| \psi_p \right|^2 \right). \label{eq:omegap}
\end{equation}
This is non-negative, because $G_{\left[ \Phi \right]}$ is given by convolutions with the positive Green's function.  By \cref{eq:apt3}, we can write
\begin{equation}
\psi_p = \frac{1}{\sqrt{2 \pi \tau}} i \Lambda \left( \sigma_p \wedge \partial_\nabla \Phi \right) - \sqrt{2} \: \del^*_\nabla \xi_p + f_p \Phi. \label{eq:psip}
\end{equation}

Applying the bounds in \Cref{thm:bnt,thm:tan} to \cref{eq:omegap,eq:psip}, we obtain
\begin{equation}
\frac{1}{i} \Omega_{\left[ \vor \right]} \left( X_p, i X_p \right) = \frac{1}{2 \pi \tau} G_{\left[ \Phi \right]} \left( \left| \sigma_p \right|^2 \left| \partial_\nabla \Phi \right|^2 \right) + G_{\left[ \Phi \right]} \left( O \left( \exp \left( - \frac{\sqrt{\tau} \dist_p}{c} \right) \right) \right).  \label{eq:cur1}
\end{equation}

By \cref{ineq:green1}, and the positivity of the Green's function the last term is
\begin{equation}
G_{\left[ \Phi \right]} \left( O \left( \exp \left( - \frac{\sqrt{\tau} \dist_p}{c} \right) \right) \right) = O \left( \tau^{-1} \exp \left( - \frac{\sqrt{\tau} \dist_p}{c} \right) \right).  \label{eq:cur3}
\end{equation}

\Cref{thm:bnt} and \cref{eq:bra} gives us
\begin{equation}
H_{\left[ \Phi \right]} \left( \chi_p w \right) = \frac{1}{2} \chi_p \left| \partial_\nabla \Phi \right|^2 + O \left( \tau \exp \left( - \frac{\sqrt{\tau} \dist_p}{c} \right) \right).  \label{eq:chiw}
\end{equation}
Thus we can write the main term in \cref{eq:cur1} as
\begin{align}
G_{\left[ \Phi \right]} \left( \frac{\left| \sigma_p \right|^2 \left| \partial_\nabla \Phi \right|^2}{2 \pi \tau} \right) &= G_{\left[ \Phi \right]} \left( \frac{\chi_p \left| \partial_\nabla \Phi \right|^2}{2 \pi \tau} \right) + O \left( G_{\left[ \Phi \right]} \left( \tau \dist_p^2 \exp \left( - \frac{\sqrt{\tau} dist_p}{c} \right) \right) \right)  \\
&= \chi_p \frac{w}{\pi \tau} + O \left( \tau^{-1} \exp \left( - \frac{\sqrt{\tau} \dist_p}{c} \right) \right), \label{eq:cur2}
\end{align}
since $\left| \sigma_p \right|^2 - \chi_p = O \left( \dist_p^2 \right)$ by \eqref{eq:sigmabound}.  Combining \cref{eq:cur1,eq:cur2,eq:cur3} yields
\begin{equation}
\frac{1}{i} \Omega_{\left[ \vor \right]} \left( X_p, i X_p \right) = \chi_p \frac{w}{\pi \tau} + O \left( \tau^{-1} \exp \left( - \frac{\sqrt{\tau} \dist_p}{c} \right) \right).
\end{equation}
This completes the proof of \cref{eq:cb,eq:coeffbound}.
\end{proof}

\smallskip

\section{The Asymptotic Berry Holonomy}
\label{sec:hol}

A connection on a principal $G$-bundle $P \rightarrow X$ defines the notion of parallel transport; cf. \cite{KN69}*{Chapter~II}.  Parallel transport around a loop is called holonomy.  Holonomy can be viewed as a map from the loop space of $X$ to the space of conjugacy classes of $G$.  For Abelian $G$, the later space is canonically isomorphic to $G$.

\smallskip

In our case, the $\tau$-vortex principal bundle, $\P_\tau \rightarrow \M_\tau$, is a principal $\G$-bundle equipped with the Berry connection.  The physical interpretation is that if one adiabatically moves the divisor points along a curve $\Gamma$ in $\Sym^d \left( \Sigma \right)$, then the corresponding $\tau$-vortex field evolves by the parallel transport defined the Berry connection (cf.  \cite{K50}, and \cite{B84}).  In particular, when $\Gamma$ is a loop, the holonomy of the Berry connection, called the {\em Berry holonomy}, is a gauge transformation.  In this section, we give analytic and topological descriptions of the gauge transformations that arise as Berry holonomies.

\smallskip

Since the Berry holonomy is a map from the loop space of $\tau$-vortex moduli space, we recall some well-known properties of loops in $\M_\tau \cong \Sym^d \left( \Sigma \right)$.

We call a loop $\Gamma$ in $\Sym^d \left( \Sigma \right)$ a {\em single vortex loop} if only one of the divisor points moves, and all other divisor points are fixed.  In other words, single vortex loops are induced by loops in $\Sigma$ that are based at one of the divisor points.  Every loop in $\Sym^d \left( \Sigma \right)$ can be decomposed up to homotopy (and thus homology) to a product of single vortex loops. Moreover,
\begin{equation}
H_1 \left( \M_\tau; \Z \right) \cong H_1 \left( \Sym^d \left( \Sigma \right); \Z \right) \cong H_1 \left( \Sigma; \Z \right),  \label{eq:hur}
\end{equation}
where the last isomorphism is given by sending single vortex loops to their homology classes by the Hurewicz homomorphism.

\smallskip

Recall that the complement of $\Sym^d_s \left( \Sigma \right)$ is called the big diagonal. A loop in $\Sym^d \left( \Sigma \right)$ is {\em regular} if it is a smooth, embedded (immersed, if $d = 1$) loop that does not intersect the big diagonal.  The big diagonal is empty when $d = 1$. When $d > 1$ the big diagonal is a subvariety of codimension at least 2, thus every smooth loop in $\Sym^d \left( \Sigma \right)$ can be made regular after a small smooth perturbation.  Now consider the canonical covering map $\Sigma^{\times d}_s \rightarrow \Sym^d_s \left( \Sigma \right)$, where $\Sigma^{\times d}_s$ is the space of ordered $d$-tuples in $\Sigma$ without repetition.  Given a regular loop $\Gamma$ that starts at the simple divisor $D = \Gamma \left( 0 \right) \in \Sym^d_s \left( \Sigma \right)$, each lift $\widetilde{D} \in \Sigma^{\times d}$ of $D$ determines a unique lift $\widetilde{\Gamma}$ of $\Gamma$.  The lift $\widetilde{\Gamma}$ can be regarded as a $d$-tuple $\left( \gamma_1, \ldots, \gamma_d \right)$ of curves (not necessarily loops) in $\Sigma$.  The {\em shadow of} $\Gamma$, $\rm{sh} \left( \Gamma \right)\subset\Sigma$ is
\begin{equation}
\rm{sh} \left( \Gamma \right) = \bigcup\ \img \left( \gamma_i \right), \label{eq:loops}
\end{equation}
where the union is over all non-constant $\gamma_i$.  The set $\rm{sh} \left( \Gamma \right)$ has a natural orientation coming from the orientation of $\Gamma$.  Since $\Gamma$ is regular, $\rm{sh} \left( \Gamma \right)$ is a union of immersed, oriented loops, hence it is an integer 1-cycle in $\Sigma$.  The homology class in $H_1 \left( \Sigma; \Z \right)$ represented by $\rm{sh} \left( \Gamma \right)$ is independent of the choice of the lift $\widetilde{D}$.  We denote this class by $\left[ \Gamma \right]$.  It is easy to check the homotopy class of $\Gamma$ in $\Sym^d \left( \Sigma \right)$ is sent to the homology class $\left[ \Gamma \right]$ by the \hyperref[eq:hur]{isomorphism \eqref{eq:hur}}.  In general, for a single vortex loop, only one of the $\gamma_i$'s is not constant, say $\gamma$, and $\left[ \Gamma \right] = \left[ \gamma \right] \in H_1 \left( \Sigma; \Z \right)$.

\smallskip

\begin{example} An example of $\rm{sh} \left( \Gamma \right)$ is seen on \Cref{fig:1loop}, where
\begin{equation}
\left[ \Gamma \right] = [\gamma_1] + [\gamma_2] + [\gamma_3] = [\gamma_3],
\end{equation}
since both $\gamma_1$ and $\gamma_2$ are null-homologous. Thus $\Gamma$ is homologous to a single vortex loop.

We call a loop a {\em (positively oriented) vortex interchange} if, as in \Cref{fig:vorint}, only two $\gamma_i$'s, say  $\gamma_1$ and $\gamma_2$, are not constant,  and the composition $\Gamma = \gamma_1 * \gamma_2$ is the (oriented) boundary of a disk.

\smallskip

\begin{minipage}{.48\textwidth}
 \centering
  \labellist
   \small\hair 2pt
   \pinlabel $\gamma_1$ at 40 29
   \pinlabel $\gamma_2$ at 90 45
   \pinlabel $\gamma_3$ at 145 33
  \endlabellist
 \includegraphics[scale=.84]{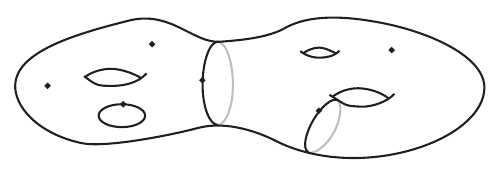}
  \captionof{figure}{Single vortex loops: One divisor point moves along one of the $\gamma_i$'s.  All other divisor points are fixed.  \label{fig:1loop}}
\end{minipage}
\begin{minipage}{.48\textwidth}
 \centering
  \labellist
   \small\hair 2pt
   \pinlabel $\gamma_1$ at 95 52
   \pinlabel $\gamma_2$ at 95 13
  \endlabellist
 \includegraphics[scale=.84]{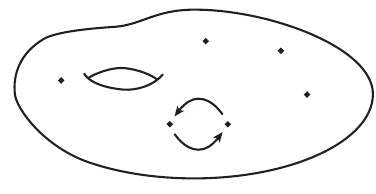}
  \captionof{figure}{Vortex interchange: One divisor point moves along $\gamma_1$ and another divisor point moves along $\gamma_2$.  All other divisor points are fixed.  \label{fig:vorint}}
\end{minipage}

\end{example}

\smallskip

Since the Berry holonomy has values in the gauge group $\G$, we also recall a couple well-known properties of gauge transformations.  Elements $g \in \G$ represent classes in $H^1 \left( \Sigma; \Z \right)$ as follows: For a closed manifold $X$ and a finitely generated Abelian group $G$, $H^n \left( \Sigma; G \right)$ is canonically isomorphic to the space $\left[ X, K \left( G, n \right) \right]$ of homotopy classes of continuous maps from $X$ to the Eilenberg-MacLane space $K \left( G, n \right)$ (cf.  \cite{H02}*{Theorem~4.57}).  Since $K \left( \Z, 1 \right) \cong \U \left( 1 \right)$ and $\G$ is homotopy equivalent to $\left[ \Sigma, \U \left( 1 \right) \right]$, we get that $H^1 \left( \Sigma; \Z \right)$ is canonically isomorphic to $\pi_0 \left( \G \right)$, which is also a group, because $\G$ is.  In fact, if $\G_0$ is the identity component of $\G$, then $\pi_0 \left( \G \right) \simeq \G / \G_0$, and the short exact sequence
\begin{equation}
\{ 0 \} \rightarrow \G_0 \hookrightarrow \G \twoheadrightarrow H^1 \left( \Sigma; \Z \right) \rightarrow \{ 0 \} \label{eq:ses}
\end{equation}
is non-canonically split.

The isomorphism between $\pi_0 \left( \G \right)$ and $H^1 \left( \Sigma; \Z \right)$ can be understood on the (co)cycle level;  since $\Sigma$ is a closed, oriented surface, $H^1 \left( \Sigma; \Z \right)$ is canonically isomorphic to $\textnormal{Hom} \left( H_1 \left( \Sigma; \Z \right), \Z \right)$.  An element $g \in \G$ defines an element $[g] \in \textnormal{Hom} \left( H_1 \left( \Sigma; \Z \right), \Z \right)$ via
\begin{equation}
[g] \left( [\gamma] \right) = g \left( \gamma \right) = \frac{1}{2 \pi i} \int\limits_\gamma g^{-1} dg \in \Z.  \label{eq:deg1}
\end{equation}

The Berry holonomy can be viewed as a map from the loop space $\Omega \M_\tau$ of $\M_\tau$ to $\G$.  It then induces a map $\hol_*$ on the connected components:
\begin{equation}
\pi_1 \left( \M_\tau \right) \xrightarrow{\sim} \pi_0 \left( \Omega \M_\tau \right) \xrightarrow{\hol_*} \pi_0 \left( \G \right) \xrightarrow{\sim} H^1 \left( \Sigma; \Z \right).
\end{equation}
Since cohomology groups are Abelian, the above map factors down to the homology, and thus defines a homomorphism:
\begin{equation}
\hol_\star: H_1 \left( \Sigma; \Z \right) \xrightarrow{\sim} H_1 \left( \M_\tau; \Z \right) \xrightarrow{\hol_*} H^1 \left( \Sigma; \Z \right),  \label{eq:hol}
\end{equation}
where the first isomorphism is from \eqref{eq:hur}.  Using \cref{eq:deg1}, an explicit formula for $\hol_\star$ can be given as follows: if $g = \hol \left( \Gamma \right)$, then $\hol_\star \left( \left[ \Gamma \right] \right)$ evaluates on any 1-cycle $\gamma$ by
\begin{equation}
\hol_\star \left( \left[ \Gamma \right] \right) \left( \left[ \gamma \right] \right) = \frac{1}{2 \pi i} \int\limits_\gamma g^{-1} dg.  \label{eq:deg2}
\end{equation}

\smallskip

Finally, recall that a {\em $k$-current} is a continuous linear functional on $\Omega^k$.  A 1-form $a \in \Omega^1$ defines a 1-current by
\begin{equation}
C_a \left( b \right) = \int\limits_\Sigma a \wedge b.  \label{eq:curr1}
\end{equation}
Similarly, a smooth 1-chain $\gamma$ defines a 1-current by
\begin{equation}
C_\gamma \left( b \right) = \int\limits_\gamma b.  \label{eq:curr2}
\end{equation}
We say that the 1-currents in \cref{eq:curr1,eq:curr2} are the {\em 1-currents defined by $a$ and $\gamma$} respectively.

\smallskip

Now we are ready to prove our main theorem about the Berry holonomy, stated in the introduction.

\begin{proof}[The proof of the \hyperlink{main:hol}{Main Theorem}] Since every smooth path can be made regular by an arbitrarily small smooth perturbation, it is enough to check regular loops, $\Gamma$.

We prove Statement~(1) first:  Let $\vor$ be a $\tau$-vortex field corresponding to $D$, and $\widehat{\Gamma}$ be the horizontal lift of $\Gamma$ starting at $\vor \in \P_\tau$.  Since $\Gamma$ is regular, $\Gamma \left( t \right)$ is simple for all $t$, thus we can apply \Cref{thm:tan} to the velocity vector $\widehat{\Gamma}'(t) = \left( a_t, \psi_t \right)$, which is horizontal, and hence obtain
\begin{equation}
\left| a_t \right| + \left| \psi_t \right| \leqslant c \tau \sum\limits_{i = 1}^d \left| \gamma_i' \left( t \right) \right| \exp \left(- \frac{\sqrt{\tau} \dist_{\gamma_i \left( t \right)}}{c} \right).  \label{eq:tanbound}
\end{equation}
By definition of the holonomy,
\begin{equation}
g_\tau \left( \vor \right) = \vor + \int\limits_0^1 \widehat{\Gamma}' \left( t \right) dt = \left( \nabla + \int\limits_0^1 a_t dt, \ \Phi + \int\limits_0^1 \psi_t dt \right).
\end{equation}
On the other hand, by the definition of the gauge action (for Abelian groups),
\begin{equation}
g_\tau \left( \vor \right) = \left( \nabla + g_\tau dg_\tau^{-1}, \ g_\tau \Phi \right).
\end{equation}
Thus we have
\begin{equation}
g_\tau dg_\tau^{-1} = \int\limits_0^1 a_t \ dt,  \label{eq:dg}
\end{equation}
and
\begin{equation}
\left( g_\tau - 1 \right) \Phi = \int\limits_0^1 \psi_t \ dt.  \label{eq:gminus1}
\end{equation}
Let $V \subset \Sigma$ be any compact set in the complement of $\rm{sh} \left( \Gamma \right)$. Since $\dist \left( \rm{sh} \left( \Gamma \right), V \right) > 0$, \Cref{thm:bnt} shows that $\left| \Phi \right| \geqslant \tfrac{\sqrt{\tau}}{2}$ on $V$ for all large $\tau$.  Hence \cref{eq:gminus1,eq:tanbound} imply that for $x \in V$
\begin{equation}
\left| \Phi \right| \left| g_\tau - 1 \right|_x \leqslant \frac{\sqrt{\tau}}{2} \left| g_\tau - 1 \right|_x \leqslant \int\limits_0^1 \left| \psi_t \left( x \right) \right| dt \leqslant c \tau \sum\limits_{i = 1}^d \int\limits_0^1 \left| \gamma_i '\left( t \right) \right| \exp \left( - \frac{\sqrt{\tau} \dist(\gamma_i \left( t \right), x)}{c} \right) dt.  \label[ineq]{ineq:c0}
\end{equation}
Using $\left| g_\tau \right| = 1$, $\left| dg_\tau^{-1} \right| = \left| dg_\tau \right|$, and \cref{eq:dg,eq:tanbound}, we also obtain
\begin{equation}
|dg_\tau|_x \leqslant \sum\limits_{i = 1}^d \int\limits_0^1 \left| a_t \left( x \right) \right| dt \leqslant c \tau \sum\limits_{i = 1}^d \int\limits_0^1 \left| \gamma_i '\left( t \right) \right| \exp \left( - \frac{\sqrt{\tau} \dist( \gamma_i \left( t \right),x)}{c} \right) dt.  \label[ineq]{ineq:c1}
\end{equation}
Combining the last two inequalities gives
\begin{equation}
\sqrt{\tau} \left| g_\tau - 1 \right|_x + |dg_\tau|_x \leqslant c' \tau \exp \left( - \frac{\sqrt{\tau} \dist \left( \rm{sh} \left( \Gamma \right), V \right)}{c} \right),  \label[ineq]{ineq:gtaubound}
\end{equation}
for all large $\tau$, which implies that $|g_\tau - 1|$ and $|dg_\tau|$ converge to 0, uniformly on $V$, as $\tau \rightarrow \infty$. This proves Statement~(1).

\smallskip

In order to prove Statement~(2), we first assume that $\Gamma$ is a single vortex loop, for which $\rm{sh} \left( \Gamma \right)$ is an embedded loop, that bounds an embedded disk $B$ in $\Sigma$, and $D = \Gamma\left( 0 \right) = \Gamma \left( 1 \right)$ has no divisor points in the interior $B$, and let $p$ be the divisor point in $D$ that is moved by $\Gamma$.  There is a canonical embedding of $B$ into $\M_\tau$ that sends a point $x \in B$ to the divisor $x + \left( D - p \right)$.  The image of this map, $\widehat{B}$, is an (oriented) disk in $\M_\tau$, whose (oriented) boundary is $\Gamma$.  We will denote this embedding by $\pi_B : B \rightarrow \widehat{B}$.

Since $\Gamma$ is null-homotopic, $g_\tau$ is in the identity component of $\G$, and so can be written as $g_\tau = \exp \left( 2 \pi i f_\tau \right)$, where $f_\tau$ is a smooth, real function on $\Sigma$.  By Stokes' Theorem,
\begin{equation}
f_\tau = \frac{1}{2 \pi i} \int\limits_{\widehat{B}} \Omega. \label{eq:stokes}
\end{equation}
If $j$ is a path as in Statement~(2) of the Main Theorem, then $\varphi_\tau = f_\tau \circ j$.  Using \cref{eq:stokes} we see that
\begin{equation}
\varphi_\tau \left( 1 \right) - \varphi_\tau \left( 0 \right) = f_\tau \left( j \left( 1 \right) \right) - f_\tau \left( j \left( 0 \right) \right) = \frac{1}{2 \pi i} \int\limits_{\widehat{B}} \left( \Omega \left( j \left( 1 \right) \right) - \Omega \left( j \left( 0 \right) \right) \right),
\end{equation}
where $j\left( 1 \right)$ is in $B$, and $j \left( 0 \right)$ is not.  To evaluate this integral, we reparametrize using $\pi_B$.  By \Cref{thm:tan}, if $\omega_\tau$ is the pullback of the K\"ahler class of $\M_\tau$ from $\widehat{B}$ to $B$ using $\pi_B$, then
\begin{equation}
\omega_\tau = \pi \tau \omega + O \left( 1 \right).
\end{equation}
For each $x \in B$, let $w_x$ be the function $w$, defined in \cref{eq:wdef}, corresponding to the divisor $D = \pi_B \left( x \right)$, and let $d_x = \dist \left( x, \{ j \left( 0 \right),j \left( 1 \right)\} \right)$.  By \cite{HJS96}*{Lemma~1.1}, $w_x|_B$ converges to $2 \pi \delta_x$, in measure, since $D \cap B = \{ x \}$. Now using \Cref{thm:cur}, the last integral above equals to
\begin{align}
\varphi_\tau \left( 1 \right) - \varphi_\tau \left( 0 \right) &= \frac{1}{2 \pi} \int\limits_B \left( \frac{w_x \left( j \left( 1 \right) \right) - w_x \left( j \left( 0 \right) \right)}{\pi \tau} + O \left( \tau^{-1} \exp \left( - \frac{2 \sqrt{\tau} d_x}{c} \right) \right) \right) \left( \pi \tau + O \left( 1 \right) \right) \omega_x \\
&= \int\limits_B \left( \delta \left( x, j \left( 1 \right) \right) + O \left( \exp \left( - \frac{\sqrt{\tau} d_x}{c} \right) \right) \right) \omega_x + O \left( \tau^{-1} \right) \\
&= 1 + O \left( \tau^{-1} \right).
\end{align}
This implies Statement~(2) in the case where $\Gamma$ is a simple vortex loop that bounds a disk.

\smallskip

In the general case, let $I \left( 2 \ep \right)$ be the tubular $2 \ep$-neighborhood of $I = \img \left( j \right)$, where $\ep$ is small enough so that $I \left( 2 \ep \right) \cap \rm{sh} \left( \Gamma \right)$ is a single embedded arc.  Let $\Gamma^\circ$ be a single, embedded bounding loop in $\Sigma$, as in the previous case, for which $\rm{sh} \left( \Gamma \right)$ and $\rm{sh} \left( \Gamma^\circ \right)$ coincide on $I \left( 2 \ep \right)$.  Let $\Gamma_{\rm out}$ and $\Gamma^\circ_{\rm out}$ denote the parts of $\Gamma$ and $\Gamma^\circ$, respectively, for which $\rm{sh} \left( \Gamma_{\rm out} \right)$ and $\rm{sh} \left( \Gamma^\circ_{\rm out} \right)$ lie in the complement of $I \left( 2 \ep \right)$.  Similarly $\Gamma_{\rm in} = \Gamma^\circ_{\rm in}$ is their common part.  Now one can join $\Gamma^\circ_{\rm out}$ with the reverse of $\Gamma_{\rm out}$ to get a piecewise smooth loop $\Gamma^{\rm new}$, such that $\rm{sh} \left( \Gamma^{\rm new} \right)$ is disjoint from $I \left( 2 \ep \right)$, and furthermore, the loop sum $\Gamma * \Gamma^{\rm new}$ and $\Gamma^\circ$ differ only by $\Gamma_{\rm out}$ and its reverse,  thus
\begin{equation}
g_\tau \hol_{\Gamma^{\rm new}} = \hol_{\Gamma^\circ}.
\end{equation}
On $I \left( \ep \right)$ we have that $\hol_{\Gamma^{\rm new}}$ converges to 1 in the $C^1$-topology as $\tau \rightarrow \infty$.  Since the crossing formula holds for $\Gamma^\circ$, it must hold for $\Gamma$ as well.  This establishes the general case of Statement~(2).

\smallskip

In order to prove Statement~(3) we pick a finite cover $\mathcal{U}$ of $\Sigma$ by coordinate charts such that for every $\left( U, \Psi \right) \in \mathcal{U}$ the preimage of the intersection $\Psi^{-1} \left( \rm{sh} \left( \Gamma \right) \cap U \right) \subset \rl^2$ is either (i) empty, (ii) the $x$-axis, or (iii) the union of the two axes.  By using a subordinate partition of unity, it is enough to prove Statement~(3) for 1-forms that are supported inside one of these charts.  Fix one such chart $\left( U, \Psi \right) \in \mathcal{U}$, and let $b \in \Omega^1$ a 1-form with $\supp \left( b \right) \subset U$. Let us also write
\begin{subequations}
\begin{align}
\Psi^*a_\tau &= df_\tau  \label{eq:psistaratau}   \\
\Psi^* b &= A dx + B dy,  \label{eq:psistarb}
\end{align}
\end{subequations}
where $A$ and $B$ are compactly supported functions on $\rl^2$.

\smallskip

If $\Psi^{-1} \left( \rm{sh} \left( \Gamma \right) \cap U \right)$ is empty, then the support of $b$ and $\rm{sh} \left( \Gamma \right)$ are disjoint, hence
\begin{equation}
C_{\rm{sh} \left( \Gamma \right)} \left( b \right) = \int\limits_{\rm{sh} \left( \Gamma \right)} b = \int\limits_{\rm{sh} \left( \Gamma \right) \cap \supp \left( b \right)} b = \int\limits_\emptyset b = 0.
\end{equation}
On the other hand, as $\tau \rightarrow \infty$, $a_\tau \rightarrow 0$ on $\supp \left( b \right)$ by Statement~(1), and hence
\begin{equation}
\left| C_{a_\tau} \left( b \right) \right| \leqslant \max\limits_{\supp \left( b \right)} \{ \left| a_\tau \right| \} \|b\|_{L^1} \rightarrow 0,
\end{equation}
which proves Statement~(3) in the case (i).

\smallskip

Next, if $\Psi^{-1} \left( \rm{sh} \left( \Gamma \right) \cap U \right)$ is the $x$-axis, then by \cref{eq:psistarb},
\begin{equation}
C_{\rm{sh} \left( \Gamma \right)} \left( b \right) = \int\limits_{\Psi^{-1} \left( \rm{sh} \left( \Gamma \right) \cap U \right)} \Psi^* b = \int\limits_{- \infty}^\infty A \left( x, 0 \right) dx. \label{eq:cgammab1}
\end{equation}
We also have
\begin{equation}
C_{a_\tau} \left( b \right) = \int\limits_\Sigma a_\tau \wedge b = \int\limits_U a_\tau \wedge b = \int\limits_{\rl^2} df_\tau \wedge \Psi^* b. \label{eq:catau1}
\end{equation}
Again using \cref{eq:psistarb}, this becomes
\begin{equation}
C_{a_\tau} \left( b \right) = \int\limits_{\rl^2} df_\tau \wedge \left( A dx + B dy \right) = \int\limits_{\rl^2} d \left( f_\tau A dx + f_\tau B dy \right) - \int\limits_{\rl^2} f_\tau \left( \frac{\partial B}{\partial x} - \frac{\partial A}{\partial y} \right) dx \wedge dy.
\end{equation}
The first integral on the right-hand side is zero by Stokes' Theorem and the fact that $A$ and $B$ are compactly supported. By Statement~(2), we choose $f_\tau$ so that it converges to 0 when $y$ is positive and to 1 when $y$ is negative.  As $\tau \rightarrow \infty$, we then have
\begin{align}
C_{a_\tau} \left( b \right) &= \int\limits_{\rl^2} f_\tau \left( \frac{\partial A}{\partial y} - \frac{\partial B}{\partial x} \right) dx \wedge dy \\
&\rightarrow \int\limits_{- \infty}^\infty \int\limits_{- \infty}^0 \left( \frac{\partial A}{\partial y} \left( x, y \right) - \frac{\partial B}{\partial x} \left( x, y \right) \right) dy dx \\
&= \int\limits_{- \infty}^\infty A \left( x, 0 \right) dx.
\end{align}
Together with \cref{eq:cgammab1} this proves Statement~(3) for case (ii).

\smallskip

Finally, if $\Psi^{-1} \left( \rm{sh} \left( \Gamma \right) \cap U \right)$ is the union of the two axes, then again by \cref{eq:psistarb},
\begin{equation}
C_{\rm{sh} \left( \Gamma \right)} \left( b \right) = \int\limits_{\Psi^{-1} \left( \rm{sh} \left( \Gamma \right) \cap U \right)} \Psi^* b = \int\limits_{- \infty}^\infty A \left( x, 0 \right) dx + \int\limits_{- \infty}^\infty B \left( 0, y \right) dy. \label{eq:cgammab2}
\end{equation}
As before, we also have
\begin{equation}
C_{a_\tau} \left( b \right) = \int\limits_{\rl^2} f_\tau \left( \frac{\partial A}{\partial y} - \frac{\partial B}{\partial x} \right) dx \wedge dy. \label{eq:catau2}
\end{equation}
But now Statement~(2) shows that, as $\tau \rightarrow \infty$, $f_\tau$ converges to 0 in the upper left quadrant, to 1 in the upper right and the lower left quadrants, and to 2 in the lower right quadrant.  Hence by \cref{eq:catau2},
\begin{align}
C_{a_\tau} \left( b \right) \rightarrow 2 &\int\limits_{\rl_+} \int\limits_{\rl_-} \left( \frac{\partial A}{\partial y} \left( x, y \right) - \frac{\partial B}{\partial x} \left( x, y \right) \right) dy dx \\
+ &\int\limits_{\rl_+} \int\limits_{\rl_+} \left( \frac{\partial A}{\partial y} \left( x, y \right) - \frac{\partial B}{\partial x} \left( x, y \right) \right) dy dx \\
+ &\int\limits_{\rl_-} \int\limits_{\rl_-} \left( \frac{\partial A}{\partial y} \left( x, y \right) - \frac{\partial B}{\partial x} \left( x, y \right) \right) dy dx \\
= &\int\limits_{- \infty}^\infty A \left( x, 0 \right) dx + \int\limits_{- \infty}^\infty B \left( 0, y \right) dy,
\end{align}
where the last step is an elementary computation.  Together with \eqref{eq:cgammab2}, this proves Statement~(3) for the case (iii).

\smallskip

To prove Statement~(4), according to \cref{eq:deg2} and the definition of the Poincar\'e duality, we need to show that for any $\left[ \gamma \right] \in H_1 \left( \Sigma; \Z \right)$,
\begin{equation}
\hol_\star \left( \left[ \Gamma \right] \right) \left( \left[ \gamma \right] \right) = \frac{1}{2 \pi i} \int\limits_\gamma g_\tau^{-1} dg_\tau = \left[ \Gamma \right] \cdot \left[ \gamma \right].  \label{eq:homint}
\end{equation}
Since everything in \eqref{eq:homint} is homotopy invariant, we have the freedom to change $\Gamma$ by a homotopy.  Recall from \eqref{eq:hur} that $\Gamma$ can be decomposed, up to homotopy, to a product of single vortex loops; thus it suffices to check Statement~(4) for a basis of $\pi_1 \left( \Sym^d \left( \Sigma \right) \right) \cong H_1 \left( \Sigma; \Z \right)$.  We use a ``symplectic'' basis: a set of simple closed curves $\{ \alpha_i, \beta_i \}_{1 \leqslant i \leqslant \textnormal{g} \left( \Sigma \right)}$ such that $\alpha_i$ intersects $\beta_i$ transversally and positively once, and $\alpha_i \cap \beta_j = \alpha_i \cap \alpha_j = \beta_i \cap \beta_j = \emptyset$ for $i \neq j$.  There is always such a set, and $\{ [\alpha_i], [\beta_i] \}_{1 \leqslant i \leqslant \textnormal{g} \left( \Sigma \right)}$ is a basis of $H_1 \left( \Sigma; \Z \right)$, with
\begin{equation}
[\alpha_i] \cdot [\alpha_j] = 0, \qquad [\beta_i] \cdot [\beta_j] = 0, \qquad [\alpha_i] \cdot [\beta_j] = \delta_{i,j},
\end{equation}
where $\cdot$ is the homology intersection.  Denote the corresponding single vortex loops in $\Sym^d \left( \Sigma \right)$ by $\{ \widehat{\alpha}_i , \widehat{\beta}_i \}_{1 \leqslant i \leqslant \textnormal{g} \left( \Sigma \right)}$.
 
To prove Statement~(4), we need only to verify \eqref{eq:homint} for every pair in the basis.  When $\gamma \in \{ \alpha_i, \beta_i \}$, and $\Gamma \in \{ \widehat{\alpha}_j, \widehat{\beta}_j \}$ with $i \neq j$, we have by Statement~(1) that
\begin{equation}
\int\limits_\gamma i g_\tau^{-1} dg_\tau = O \left( \tau \exp \left( - \frac{\sqrt{\tau} \dist \left( \gamma, \rm{sh} \left( \Gamma \right) \right)}{c} \right) \right). \label{eq:pairing1}
\end{equation}
When $\gamma = \alpha_i$ and $\Gamma = \widehat{\alpha}_i$, for some $i$, we can chose another representative $\alpha'_i$ for $\left[ \gamma \right] = \left[ \alpha_i \right]$ that is disjoint from $\alpha_i$.  Thus, again by Statement~(1), we have
\begin{equation}
\int\limits_\gamma i g_\tau^{-1} dg_\tau = O \left( \tau \exp \left( - \frac{\sqrt{\tau} \dist \left( \alpha_i, \alpha'_i \right)}{c} \right) \right). \label{eq:pairing2}
\end{equation}
Thus all integrals in \eqref{eq:pairing1} and \eqref{eq:pairing2} converge to 0 as $\tau \rightarrow \infty$.  On the other hand, these integrals are integer multiples of $2 \pi$, so they had to be 0 for all $\tau > \tau_0$.

Finally, assume that $i = j$ and $\gamma = \alpha_i$ and $\rm{sh} \left( \Gamma \right) = \widehat{\beta}_i$, or $\gamma = \beta_i$ and $\rm{sh} \left( \Gamma \right) = \widehat{\alpha}_i$.  In order to prove the first of these cases, let us fix a small embedded segment $j$ on $\gamma = \beta_i$ that intersects $\rm{sh} \left( \Gamma \right) = \alpha_i$ once positively.  Such paths exist by the construction of the basis.  Write $g_\tau |_I = \exp \left( 2 \pi i \varphi_\tau \right)$, and so $g_\tau^{-1} dg_\tau |_{I} = 2 \pi i d\varphi_\tau$. Thus by \cref{eq:deg2},
\begin{align}
\hol_\star \left( \left[ \Gamma \right] \right) \left( \left[ \gamma \right] \right) &= \int\limits_I d \varphi_\tau + \frac{1}{2 \pi i} \int\limits_{\gamma - I} g_\tau^{-1} dg_\tau \\
&= \varphi_\tau \left( 1 \right) - \varphi_\tau \left( 0 \right) + O \left( \sqrt{\tau} \exp \left( - \frac{\sqrt{\tau} \dist \left( \rm{sh} \left( \Gamma \right), \{ j \left( 0 \right), j \left( 1 \right) \} \right)}{c} \right) \right).
\end{align}
By Statement~(2), this converges to 1 as $\tau \rightarrow \infty$.  On the other hand, $\hol_\star \left( \left[ \Gamma \right] \right) \left( \left[ \gamma \right] \right)$ is an integer, so it had to be 1 for all $\tau > \tau_0$.  The same argument can be used in the case of $\rm{sh} \left( \Gamma \right) = \widetilde{\beta}_i$ and $\gamma = \alpha_i$, which completes the proof of Statement~(4) and the \hyperlink{main:hol}{Main Theorem}.
\end{proof}

\smallskip

\begin{cor} \label{cor:htpy}
For all $\tau > \tau_0$, the space $\P_\tau$ is an infinite-dimensional vector bundle over a connected, oriented and smooth manifold without boundary, $\widehat{\M}_\tau$.  This manifold has real dimension $2 d + 1$ and is a $\U \left( 1 \right)$-principal bundle over the universal cover of $\M_\tau$.  In particular $\P_\tau$ is homotopy retracts to $\widehat{\M}_\tau$.
\end{cor}

\begin{proof}
First we will impose the Coulomb gauge:  fix a $\tau$-vortex field $\vor_0 = \left( \nabla, \Phi \right) \in \P_\tau$.  We say that $\vor = \left( \nabla', \Phi' \right)$ is in {\em Coulomb gauge with respect to} $\vor_0$ if the 1-form $a = \nabla' - \nabla$ is satisfies:
\begin{equation}
d^* a = 0.
\end{equation}
For each $\vor \in \P_\tau$ there is in fact a gauge transformation $g$ that is in the identity component of $\G$, and is unique up to constant gauge transformations (factors in $\U \left( 1 \right)$), such that $g \left( \vor \right)$ is in Coulomb gauge with respect to $\vor_0$.  A proof of this, which applies to our case too, can be found, for example, in \cite{EN11}*{Lemma~2.1}.  The set $\widehat{\M}_\tau \subset \P_\tau$, called the {\em Coulomb slice}, consisting the $\tau$-vortex fields that are in Coulomb gauge with respect to $\vor_0$ intersects each fiber.

Fix then a point $x \in \Sigma$, and require $g \left( x \right) = 1$.  Such a $g = g_\vor$ is then unique, moreover, can be written as $g_\vor = \exp \left( i f_\vor \right)$, and $f_\vor$ is also unique if one prescribes $f_\vor \left( x \right) = 0$.  Set $g_{t, \vor} = \exp \left( i t f_\vor \right).$  Then the map defined as $r \left(t, \vor \right) = g_{t, \vor} \left( \vor \right)$ is a homotopy retraction of $\P_\tau$ to the Coulomb slice.

The intersection of each fiber with the Coulomb slice is a collection of circles, due to the $\U \left( 1 \right)$ ambiguity mentioned above.  Moreover, these circles are in bijection with $\pi_0 \left( \G \right) \cong \pi_1 \left( \M_\tau \right)$.  Thus $\widetilde{\M}_\tau = \widehat{\M}_\tau / \U \left( 1 \right)$ is a $\pi_1 \left( \M_\tau \right)$-cover of $\M_\tau$, which is the universal cover if connected.

Our \hyperlink{main:hol}{Main Theorem} implies that $\P_\tau$ is connected, by the following argument: let $\vor$ and $\vor'$ be two arbitrary $\tau$-vortex fields.  Since simple divisors are dense in $\Sym^d \left( \Sigma \right)$, we can assume, that the corresponding divisors are simple.  Join the two divisors by a regular path $\Gamma_0$.  Then $\hol_{\Gamma_0} \left( \vor \right)$ is equal to $g \left( \vor' \right)$ for some $g \in \G$. If $g$ is not in the identity component of $\G$, then it represents a non-zero cohomology class $\left[ g \right] \in H^2 \left( \Sigma; \Z \right)$. Let $\gamma$ be a smooth loop based at a divisor point of $\vor$ that represents the Poincar\'e dual of $\left[ g \right]$; let $\widehat{\gamma}$ be the induced single vortex loop, and set $\Gamma = \widehat{\gamma}^{-1} * \Gamma_0$.  Now $\vor$ and $\hol_\Gamma \left( \vor \right) = \vor''$ are connected by the path in $\P_\tau$ given by parallel transport. On the other hand, $\vor''$ and $\vor$ are gauge equivalent, and the connecting gauge transformation is in the identity component of $\G$, which means that there is a path from $\vor''$ to $\vor$.  Thus $\P_\tau$ is connected, but then so is $\widehat{\M}_\tau$, which completes the proof. 
\end{proof}

\smallskip

\section{The Large Area Limit}
\label{sec:area}

Consider the \hyperref[eq:glf]{energy \eqref{eq:glf}} for the critical coupling constant $\lambda = 1$ and with $\tau = 1$.  Bradlow's criterion for the existence of irreducible vortices in this case becomes
\begin{equation}
\tau_0 = \frac{2 \pi d}{\Ar (\Sigma)} < 1,  \label[ineq]{ineq:bra}
\end{equation}
using the area with respect to the given area 2-form $\omega$. Even when \cref{ineq:bra} does not hold for $\omega$, it still holds for $\omega_t = t^2 \omega$ if $t > t_0 = \sqrt{\tau_0}$.

\smallskip

Let $P_t$ be the space of all solutions of the 1-vortex equations with K\"ahler form $\omega_t$ for $t > t_0$. A pair $\left( \nabla, \Phi \right) \in \C_L \times \Omega^0_L$ is in $P_t$ if
\begin{subequations}
\begin{align}
i \Lambda_t F_\nabla &= \frac{1}{2} \left( 1 - \left| \Phi \right|^2 \right) \label{eq:newvo1} \\
\del_\nabla \Phi &= 0,  \label{eq:newvo2}
\end{align}
\end{subequations}
where $\Lambda_t = \Lambda/t^2$.  Let $M_t$ be the corresponding moduli space $P_t / \G$.  Bradlow's Theorem still holds, hence $M_t \cong \Sym^d \left( \Sigma \right)$, where the diffeomorphism is again given by the divisor of the $\Phi$-field.

By \cite{B90}*{Proposition~5.1}, the following diagram is commutative when $t^2 = \tau$:
$$
\xymatrix{
P_t \ar[rr]^{\Psi_t} \ar[rd] & & \P_\tau \ar[ld] \\
&\Sym^d \left( \Sigma \right)},
$$
where $\Psi_t$ is the isomorphism of principal bundles given by $\Psi_t \left( \nabla, \Phi \right) = \left( \nabla, t \Phi \right)$.

The $L^2$-metric on $P_t$ is defined by \cref{eq:riem}, but with Hodge operator and area form given by $\omega_t$.  For all $X \in T P_t$, we now have:
\begin{equation}
\| \left( \Psi_t \right)_* X \|_{\P_\tau} = t \| X \|_{P_t}.  \label{eq:norm}
\end{equation}
Thus the $L^2$-metric of $P_t$ is conformally equivalent to the pullback of the $L^2$-metric of $\P_\tau$ via the bundle isomorphism $\Psi_t$.

\smallskip

The Berry connection on $P_t \rightarrow M_t$ is again defined as the orthogonal complement of the vertical subspaces, hence it is the same as the pullback of the Berry connection on $\P_\tau$ via $\Psi_t$.  Thus the results of \Cref{thm:tan,thm:cur} and our \hyperlink{main:hol}{Main Theorem} hold in the {\em large area limit} $\left( t \rightarrow \infty \right)$:

\smallskip

\begin{main2}
The conclusions of the \hyperlink{main:hol}{Main Theorem} in the introduction hold for the principal $\G$-bundle $P_t \rightarrow M_t$ with $\tau$ replaced everywhere by $t^2$.
\end{main2}

\bibliography{references}

\end{document}